\documentclass[10pt,journal,letterpaper,compsoc]{IEEEtran}
\usepackage{setspace}

\usepackage{verbatim}
\usepackage{amsmath}
\usepackage{amssymb}

\usepackage{graphicx}
	\graphicspath{{Figures/}}

\usepackage[noadjust]{cite}
\usepackage{subfigure}
\usepackage{multirow}
\usepackage{enumerate}
\usepackage{ifthen}

\newcommand{\checkOneCol}{\ifthenelse{\numcol = 1}}

\newcommand{\numcol}{2}
\renewcommand{\iff}{\ensuremath{\mathrm{iff}}}

\newtheorem{lemma}{Lemma}

\begin{document}
\title{Cooperation Optimized Design for Information Dissemination in Vehicular Networks using Evolutionary Game Theory}
\author{Abhik Banerjee, Vincent Gauthier, Houda Labiod, Hossam Afifi
\IEEEcompsocitemizethanks{\IEEEcompsocthanksitem Abhik Banerjee is with the Department RST, CNRS SAMOVAR, Telecom SudParis, Institut Mines Telecom. 9 rue charles Fourier, 91290 Evry, France. E-mail:  abhik.banerjee@ieee.org
\IEEEcompsocthanksitem V. Gauthier, and H. Afifi are with the Department RST, CNRS SAMOVAR, Telecom SudParis, Institut Mines Telecom. 9 rue charles Fourier, 91290 Evry, France.\protect\\ E-mail: \{vincent.gauthier, hossam.afifi\}@telecom-sudparis.eu
\IEEEcompsocthanksitem Houda Labiod is with the Department INFRES, CNRS LTCI, Telecom ParisTech, Institut Mines Telecom. E-mail: labiod@telecom-paristech.fr
\IEEEcompsocthanksitem To whom correspondence should be addressed.\protect\\ E-mail: vincent.gauthier@telecom-sudparis.eu}
\thanks{}
}
	
%
%
%
%
%
%
%
\IEEEcompsoctitleabstractindextext{%
\begin{abstract}
We present an evolutionary game theoretic approach to study node cooperation behavior in wireless ad hoc networks. Evolutionary game theory (EGT) has been used to study the conditions governing the growth of cooperation behavior in biological and social networks. We propose a model of node cooperation behavior in dynamic wireless networks such as vehicular networks. Our work is motivated by the fact that, similar to existing EGT studies, node behavior in dynamic wireless networks is characterized by decision making that only depends on the immediate neighborhood. We adapt our model to study cooperation behavior in the context of information dissemination in wireless networks. We obtain conditions that determine whether a network evolves to a state of complete cooperation from all nodes. Finally, we use our model to study the evolution of cooperation behavior and its impact on content downloading in vehicular networks, taking into consideration realistic network conditions.
\end{abstract}

}

\maketitle

\IEEEdisplaynotcompsoctitleabstractindextext

\IEEEpeerreviewmaketitle

%
%
%
%
%
%
%

\section{Introduction}\label{sec:intro}
Packet delivery in multi-hop wireless networks depends critically on the interaction between individual nodes. As successful transmission of a packet from a source to a destination requires forwarding by intermediate nodes, it is necessary to understand how node cooperation can be engendered and sustained. This becomes more crucial when considering future deployments of multi-hop wireless networks such as vehicular networks which are characterized by high scalability and dynamicity. 

In this paper, we are interested in understanding the evolution of node cooperation behavior in a wireless network. We are particularly interested in studying how node cooperation behavior evolves in dynamic scenarios such as that of a vehicular networks. Our primary motivation is that node decision making in such contexts needs to be spontaneous and myopic in nature, since the available information is limited to locally available information. Our approach is inspired from existing studies in Evolutionary Game Theory which study evolution of node behavior in similar contexts.

Our motivations stem from the fact that the current trend in game theoretical approach in wireless networks mostly focus on optimization issues based on technical criteria \cite{Chiasserini2003} (e.g. minimization of the energy spending, and traffic maximization). However, in practice the cooperative behavior might not be only driven by algorithmic incentives but also driven by user's behaviors that might just decline to participate in the information relaying process as in peer to peer networks where some users ("Free-Riders"), choose not to engage for themselves in the forwarding process. Beyond this decision based on the local behavior of nodes, there is an intricate relationship between the ratio of cooperator in the network and the information dissemination process characterized by a phase transition (e.g. too few forwarders will lead to a very restricted diffusion). To face this issue the network as a whole needs to maintain the sufficient conditions for diffusion to take place. This collective behavior need to be addressed as well as the local behavior of a node that decides or not to participate to the forwarding process regardless if the decision process is taken base on technical, or user centric concern or a mix of both. The Evolutionary Game Theory seems perfectly suited to address these issues: first because the outcome of the game is derived from the local decision and secondly by the social pressure could be expressed as the ratio of people sharing the same behavior.

\subsection{Evolutionary Game Theory}\label{subsec:egt-ovw}
Evolutionary game theory \cite{Nowaka,Weibull1997} is a general theoretical framework that can be used to study many biological problems including but not limited to host-parasite interactions \cite{Nowak1994}, eco-systems \cite{Friedman1991}, animal behavior \cite{Smith1973}, social evolution, and human language \cite{Nowak2002}. Evolutionary games arise to study population dynamic where the fitness of individuals is not constant, but depends on the relative abundance of strategies in the population pool. Originally intended for studying cooperative interactions in biological systems, EGT studies how the strategy adopted by a node evolves as a result of its interactions with others. The framework  have recently included feature like stochastic dynamics, finite population size, and structured populations. Indeed, the evolution of network characteristics as a function of interactions between individual nodes has become the primary focus of recent research on Evolutionary Game Theory (EGT) \cite{Lieberman2005,Ohtsuki2006} (cf. Fig. \ref{fig:intro} for a example of EGT with a structured population). A salient property to notice is EGT doesn't suppose that all of the players make rational choices (i.e. a player in a EGT does not necessary adapt his strategy as function of his opponents' strategies, or don't act only toward their self-interest) as oppose to classical game theory. As result in EGT we are more interested in studying the condition for achieving an evolutionary stable state (ESS) which is akin of Nash equilibrium \cite{Nowaka}. The survival of a strategy in a network depends on the benefits achievable by it in comparison to other strategies. It is a well known result from EGT that in unstructured populations (random encounter between nodes), natural selection favors defection over cooperation, as opposed with model with a structured population (i.e. on graphs) where cooperation is favored over defection \cite{Nowak2004a, Ohtsuki2006, Lieberman2005, Santos2008}. Indeed, in structured population model where the interactions are constrained either by spatial or social relationships the
emergence of cooperative behavior is favored when the altruistic acts exceed the connectivity of a node, thereby making the cooperation on graph a valuable option in case of mild connectivity. This fact is especially true in the case of scale free graph \cite{Barabasi1999}. Another important aspect is the problem of how the survival of cooperation in a social system depends on the mobility of nodes in the graph, in \cite{Meloni2009} authors show that cooperation can survive given that both the temptation to defect and the velocity at which agents move are not too high. In \cite{Roca2011} authors provide new insight on a similar problem but in the case of Public Good Games instead of the Prisoner's Dilemma. 

\subsection{Motivation for Using Evolutionary Game Theory to Study Node Cooperation Behavior}\label{subsec:motiv-egt}
Traditional game theory based approaches to improve packet delivery in wireless networks typically involve design of utility functions that depend on multiple parameters. In such scenarios, arriving at an optimal configuration either involves sufficient amount of available information or learning over time.

Existing research on improving cooperation among wireless nodes seeks to do so either through incentive and punishment mechanisms or by design of detailed utility functions. While extensive literature exists on the subject of game theory in wireless multi-hop networks, we draw attention to a few that are similar in motivation to ours. Ng et al. proposed a architecture in \cite{NgPunishment} in which monitor nodes impose punishments on selfish nodes to induce cooperation. Focusing on primarily on unicast packet routing in vehicular ad hoc networks (VANETs), Chen et al. \cite{ChenCoalGm} propose an incentive mechanism based on coalitional game theory which rewards nodes that take part in packet forwarding. Concentrating also on VANETs, Schwartz et al. studied utility function based approaches for improving data dissemination in \cite{schwartz2011analysis}. The authors analyze utility functions which incorporate message characteristics such as priority, distance to its location and age along with mobility information of the vehicle such as route information. A data selection algorithm based on similar utility functions was proposed in \cite{schwartz2012achieving} with a view to achieving fair distribution of data among vehicular nodes. Shrestha et al. consider separate utility functions for Road Side Units (RSUs) and On Board Units (OBUs) and propose an algorithm based on Nash Bargaining Solution to achieve large scale data dissemination in VANETs \cite{shrestha2008wireless}.

Contrasted to existing literature, we are interested in studying network scenarios where nodes do not possess enough information to construct detailed utility functions. Our objective here is to focus on how node behavior evolves in scenarios in which, not only nodes have limited information about the network, but their interaction with partners nodes also change frequently. In such a scenario, it seems intuitive for nodes to make myopic decisions based on instantaneously available information. In the absence of global network information, it is plausible for nodes to choose strategies which may only seem to be the best, based on the limited information available, but are not optimal from the standpoint of the entire network. The attractiveness of EGT for vehicular networks stems from its ability to adapt to various network parameters that change over time. Further, there is a greater need for cooperation among nodes in a vehicular network for applications requiring dissemination of information to all nodes. In such scenarios, while node cooperation is vital to ensuring maximum coverage for the disseminated information, the rational choice for any node is to receive such information but not to cooperate by forwarding it, since there are no gains to be received by doing so. Compared to classical game theory, an interesting aspect of EGT is that it provides a way to understand how node strategies can thrive even when such behavior is not rational in terms of received payoffs. 

A topic of particular interest for the study of cooperation is the Public Goods Game (PGG), in which the contributions from cooperating nodes are distributed equally among all nodes in the group. The rational behavior for an individual is to defect and not contribute anything to the public good, resulting in the tragedy of commons \cite{hardin2009tragedy}. However, despite defection being the Nash Equilibrium, it is observed that cooperation can not only survive in a wide range of situations but even dominate at times. The conditions for cooperation in PGGs has seen extensive research focus \cite{hauert2004dynamics,hauert2006evolutionary,cardillo2012velocity}. However, to our knowledge, PGGs have not been studied in the context of various wireless ad hoc network deployments such as VANETs. PGGs are suitable for the study of node behavior in wireless networks due to the prevalence of group interactions. Further, as with PGGs, nodes in wireless networks are presented with a natural incentive to defect.

\begin{figure*}
\centering
\includegraphics[width=6in]{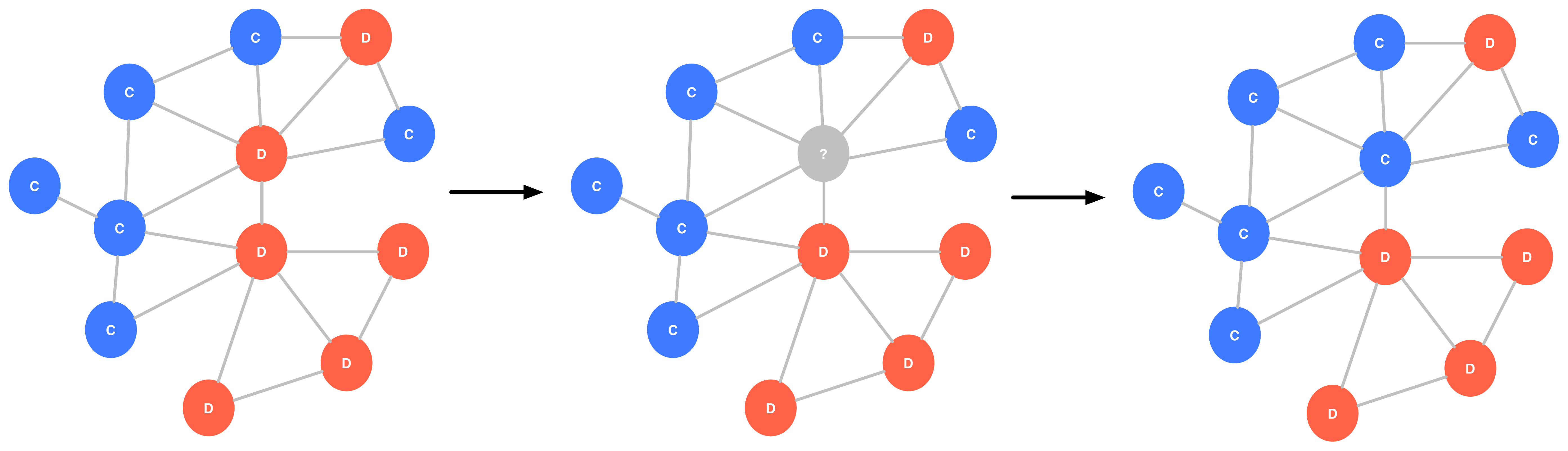}
\caption{Example of Evolutionary Game in a structured population. The rules of the game. Each individual occupies the vertex of a graph and derives a payoff according to the game rules and from interactions with adjacent individuals. In a next step the selected node change it behavior according to the outcome of the previous game played.}
\label{fig:intro}
\end{figure*}

\subsection{Contributions}\label{subsec:contrib}
In this paper, we seek to understand node cooperation behavior in wireless networks using evolutionary game theory. We first propose a public goods games (PGG) framework for wireless networks. The key feature of our proposed model is to adapt existing designs of PGGs to the  unique constraints of a wireless network. Specifically, our model takes into consideration two important limitations of wireless networks. Firstly, we consider that nodes only posses information from the immediate neighborhood. Secondly, contrasted to existing definitions of PGGs, benefits of cooperation are also realized only from one hop neighbors. Conversely, we also make a crucial observation that wireless multi-hop networks have an intrinsic ability to boost cooperation due to implicit sharing of information as a result of wireless broadcast advantage \cite{wieselthier2000construction}. Subsequently, we use our model to study node cooperation behavior for the problem of information dissemination. Finally, we focus on a realistic problem setup of content downloading in vehicular networks and study how node cooperation behavior evolves in such a scenario. 
 
\subsection{Related Work}\label{subsec:rel-work}
Sharing a similar motivation as ours, existing research on wireless networks has seen some focus on algorithm design and analysis based on evolutionary game theory. In \cite{tembine2009evolutionary}, Tembine et al. used evolutionary dynamics to analyze node behavior over time in the context of two problems, namely channel access in slotted ALOHA systems and decentralized power control. Evolutionary coalition games were used to study the problem of network formation in \cite{tembine2010evolutionary} and \cite{khan2011evolutionary} while dynamic routing games were studied in \cite{tembine2011dynamic}. In their work  in \cite{cheng2011ecology}, Cheng et al. studied the evolution of misbehavior among secondary users in a cognitive radio network. 

A key aspect of our work which distinguishes it from existing research is the fact that we study node cooperation behavior in the context of information dissemination in wireless networks. Our approach aims to identify the mutual impact that the evolution of node cooperation and the information dissemination process have on each other. As with the problem of information dissemination in vehicular networks considered here, there exist other scenarios in which free-riding is the only rational strategy for nodes to follow. Indeed, existing research has explored the use of evolutionary game theoretic techniques to study node cooperation in scenarios such as peer-to-peer file sharing \cite{sasabe2007caching,sasabe2010user,matsuda2010evolutionary}. However, the distinguishing feature of our work is that our model of node behavior is particularly suited for wireless network scenarios since we take into consideration aspects like node interactions limited to immediate neighborhood and the intrinsic ability of the wireless medium to boost cooperation.

%
%
%
%
%
%
%

\section{Public Goods Games}\label{sec:PGG}
\subsection{Overview}\label{subsec:pgg-ovw}
Public Goods Games (PGG) is a model of interaction for groups of individuals, all of which are interested in reaping benefits from a shared public good. In a PGG, the benefit available to a particular group is divided equally among all its members. Cooperating members contribute a certain cost into the public pool. The total contributions are then multiplied by a synergy factor $r$ (discussed next) and divided equally among each individual in the group, irrespective of whether it is a cooperator or a defector. The net payoff of a cooperator, therefore, is the benefit received from the PGG minus the cost contributed. A node $i$ in a network takes part in multiple PGGs determined by the size of its neighborhood. For each PGG centered on each neighbor of node $i$, the benefit received by $i$ is determined by the number of cooperators in each such group.

An important component of the PGG is the factor $r$, termed as the synergy factor, used to account for the synergistic effects of cooperation. Synergistic effects of cooperation refer to the fact that the benefits received as a result of cooperation can be greater than the sum of the contributed costs. The synergy factor $r$ is a measure of the benefits obtained due to cooperation. Increased synergistic effects imply greater benefits received by individual nodes compared to the costs incurred and is, therefore, more likely to induce cooperative behavior \cite{Aviles1999,Kun2006}. 

Depending on the amount contributed by each cooperator, two variants of PGGs have been studied in the literature. In the simplest case, a cooperator contributes a fixed cost $c$ to each PGG it is involved with. For a group with $N$ participants and $n_C$ cooperators, the payoff of a cooperator and a defector are, respectively,
\begin{align}
\pi_C &= r \frac{n_C}{N} c - c \nonumber \\
\pi_D &= r \frac{n_C}{N} c . \nonumber
\end{align}
Each node accumulates the payoffs obtained from all its neighboring PGGs. 

A second variant of PGG considers the scenario in which a cooperator does not contribute a fixed amount $c$ for each PGG they are involved with. Instead, a cooperator $i$ reserves a total amount $c$ to all the PGGs it is involved with. Let $k_i$ denote the degree of node $i$, i.e. the number of immediate neighbors of $i$. Since, for a node $i$ with degree $k_i$, there are a total of $(k_i + 1)$ PGGs, the contribution made to a single PGG is $\frac{c}{k_i + 1}$. Subsequently, the cooperator and defector payoffs for a node $i$ from a PGG centered at node $j$ are obtained as,
\begin{align}
\pi_C &= \frac{r}{k_j + 1}\sum\limits_{x \in \mathcal{N}_j} \frac{c}{k_x + 1} {s_x} - \frac{c}{k_i + 1} \nonumber \\
\pi_D &= \frac{r}{k_j + 1}\sum\limits_{x \in \mathcal{N}_j} \frac{c}{k_x + 1} {s_x} \nonumber
\end{align}
where $s_x = 1$ if $x$ is a cooperator and $0$ otherwise, while $\mathcal{N}_j$ denotes the set of neighbors of $j$.

Based on the payoffs obtained, nodes update their strategy at each time instant by comparing their payoff values to a randomly chosen neighbor. If the chosen neighbor $j$ of a node $i$ has a higher payoff, $i$ adopts $j$'s strategy with a probability given by the following function:
\begin{equation}\label{eq:stg-upd}
P_{ij} = \frac{1}{1 + exp[(\pi_i - \pi_j)/\kappa]}
\end{equation}
where $\kappa$ denotes the intensity of selection and is usually set to $1$.

\subsection{Motivation}\label{subsec:motive-pgg}
Our choice of PGGs as the evolutionary game to formulate the above mentioned information dissemination problem is motivated by direct analogies that exist between the two. As with any peer-to-peer system, one of the primary challenges of the information dissemination problem lies in the tendency of nodes to free-ride over the contributions of others. As individual nodes are only interested in receiving information, they can choose to not act as forwarding nodes. However, while free-riding is the more profitable behavior for individual nodes, the actual algorithm performance crucially depends on the forwarding behavior of nodes. A network consisting of only free-riders soon results in no packets being received by any node.  Such behavior directly correlates to node behavior in a public goods game, in which defection promises higher payoffs than cooperation and therefore, is the more profitable strategy. However, a population of all defectors soon dies off with zero payoffs as no contributions are made. Despite the clear benefits of defection, however, cooperation is seen to thrive and even dominate in a wide range of scenarios. Existing studies on evolutionary games, thus, offer a way to identify application scenarios in vehicular networks in which cooperation behavior among nodes can be expected to thrive and those in which it does not.

%
%
%
%
%
%
%

\section{A Public Goods Framework for Multi-hop Wireless Networks}\label{sec:pggfwk-winet}
As outlined earlier, a public goods game presents an attractive method of accurately modeling node cooperation behavior for the information dissemination problem. However, the nature of packet transmissions in a multi-hop wireless network imposes restrictions which limit the amount of payoff available to a node. As a result, the classical definition of PGG cannot be used directly. In order to understand how node cooperation behavior in a multi-hop wireless network evolves over time, we formulate a public goods framework taking into consideration the constraints imposed by a wireless network in determining node payoffs. Subsequently, we use simulation results to identify the impact of increase in the synergy factor and node mobility on the level of cooperation in a network. The results presented in this section and the next focus on the node cooperation behavior in the steady state. The steady state refers to the long term behavior to which the nodes in the network converge to.

\subsection{Adapting PGG to Multi-hop Wireless Networks}\label{subsec:formul-pggwinet}
Our approach aims to define the payoff received by a node in a single time slot. We use the second variant of PGG defined earlier as a baseline since it incorporates diversified contributions from each cooperator. Such a design is suited especially for broadcast protocols in wireless networks since a single transmission from a node is sufficient to reach all its neighbors. The contribution of a cooperating node, thus, does not grow with the number of neighbors, but stays constant. Considering the transmission of a packet within a single time slot as the unit cost, this is the maximum cost a cooperator can contribute in a time slot. However, the likelihood of a node transmitting in a slot is determined by the contention in its neighborhood. The contribution of a node, thus, depends on its degree. We express the contribution of a cooperator $i$ as $\frac{c}{{k_i}+1}$, where $k_i$ is the degree of $i$ and $c$ is the constant cost per contribution, taken equal to $1$ in this paper since we consider the transmission in a time slot to be of unit cost.

A second distinguishing factor is that, as any transmission from a node can only be received by its immediate one hop neighbors, a cooperator only contributes to the game centered on itself. In the traditional definition of PGG, a node takes part in games centered on all nodes in its neighborhood and subsequently receives benefits contributed by cooperators in each of the games. Such a model is not suitable for multi-hop wireless network since contributions from nodes more than one hop away cannot be immediately realized by a node. Moreover, in any time slot, a node only makes a single contribution. In our game formulation, the benefits received by a node $i$ consist only of the contributions made by each cooperating neighbor $j$ for the game centered on itself. We express the payoff received by a node $x$ at any time slot in terms of the traditional definition of PGG as,
\begin{equation}\label{eq:pi-pggfrmwk}
\pi_x = r \sum\limits_{y \in (\mathcal{N}_x \cup x)}\sum\limits_{z \in (\mathcal{N}_y \cup y)} \frac{c}{{k_y}+1} {q_y} - \sum\limits_{y \in (\mathcal{N}_x \cup x)} \frac{c}{{k_x}+1} {q_x}
\end{equation}
Here, the inner summation of first term and the summation of the second term correspond to the games centered on all the neighbors of a node. The distinction from the traditional definition of PGG is made by the variable $q_i$ which takes the value $1$ only if $i$ is a cooperator for a game centered on itself, and $0$ otherwise. 

Note that the benefits are not divided by the neighborhood size due to the fact that only a single game is played. Here, we make the observation that, as a single transmission is sufficient for sharing information with all nodes within a neighborhood, synergistic effects exist implicitly in a wireless network. The sum of the benefits received by all nodes in a neighborhood is always likely to exceed that of the total cost incurred. Thus, the synergy factor $r$ in the above formulation represents the additional benefits that can be realized depending on network conditions. In this context, we envision $r$ as representing benefits achievable to a node in addition to the number of received packets. Thus, $r$ may be used to quantify other aspects of network flow such as traffic priority in terms of either the importance of the information content or prioritization with regard to network costs.

\begin{figure}
\centering
\includegraphics[width=0.9\columnwidth]{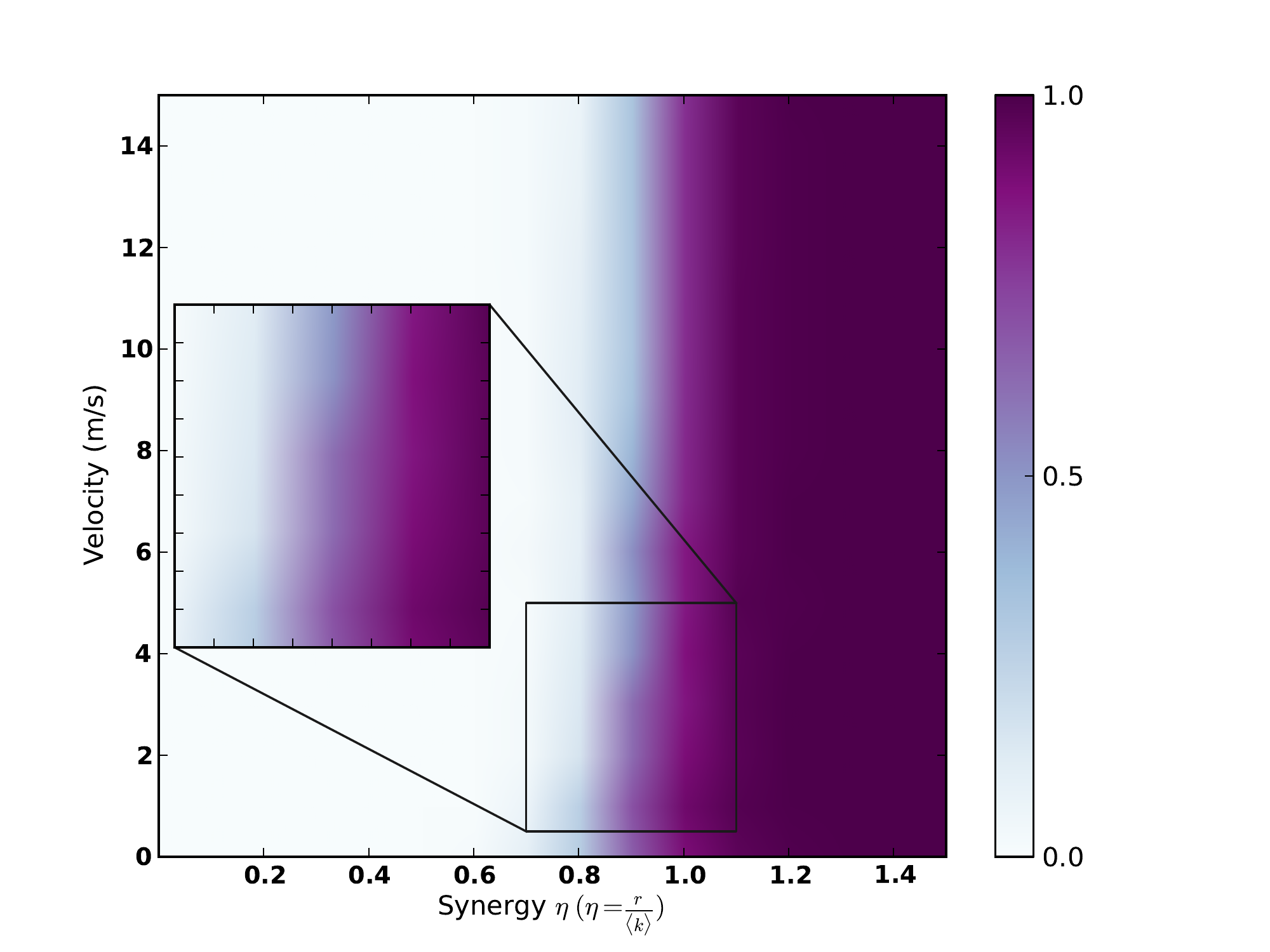}
\caption{Impact of node velocity and synergy factor on the fraction of cooperators in the steady state.}
\label{fig:pggfwk-rv}
\end{figure}

\subsection{Evolution of Node Cooperation Behavior}\label{subsec:simres-fmwk}
We obtain simulation results to identify the impact of the synergy factor and node mobility on the level of cooperation. Fig. \ref{fig:pggfwk-rv} shows the expected fraction of cooperators in the steady state. The results are presented  for increasing values of the synergy factor $r$ normalized by the average neighborhood size in the network $(\langle k \rangle + 1)$.

As with most existing studies on public goods games, cooperation is seen to dominate for high values of the synergy factor. This results form a high benefit to cost ratio, due to which the cost incurred by cooperators becomes negligible to the overall payoff received. We note that cooperation can be sustained when the normalized synergy factor $\eta \geq 0.6$ for a static network. This is similar to existing results on PGGs with diversified contributions from cooperators. However, in PGG formulation adapted to wireless networks above, two opposing factors are present. As nodes only play a single game in their neighborhood, this acts as a mitigating factor from the point of view of enhancing cooperation as nodes do not accumulate benefits from games centered on other nodes in their neighborhood. This mitigating factor, though, is countered by the implicit synergistic effects existing in a wireless network, which act to increase the benefits available to a node.

Increase in node velocity shifts the domination of cooperative behavior to higher values of synergy factor. This results from the fact that, for low to moderately high values of synergy, node cooperation behavior is more likely to survive in neighborhoods that have a higher fraction of cooperators. For increasing node velocities, neighborhood structures change frequently leading to breakdown of clusters of cooperators. This results in a higher likelihood of defection emerging as the dominating strategy. This effect vanishes at higher values of the synergy factor since the cost component of the payoff becomes negligible compared to the benefit received by a node. As a result, cooperators are less likely to switch to defection upon comparing payoffs with neighborhood defectors. The result in figure \ref{fig:pggfwk-rv} show a sharp shift from a network state where almost all the nodes have a tendency to defect to a state where almost all of them cooperate. This is suggesting that the response of the model has a tipping point, when it is passed it indicate that the network will induce propel a change in behavior into the network. This is also suggesting that once the change in behavior is triggered the system become resilient to local change in the network structure and local behavior shift. As consequence the evolution the of cooperation in the model is only being lightly affected by the velocity of node and the rate of change in the network topology.

%
%
%
%
%
%
%

\section{Public Goods Games for Information Dissemination in Multi-hop Wireless Networks}\label{sec:pgg-infodissm}
Using the basic framework of public goods games defined in the previous section, we now formulate a PGG to model information dissemination. Efficient delivery of data packets in a multi-hop network depends critically on the set of intermediate nodes which are willing to act as forwarding nodes. In dynamically changing networks such as vehicular networks, this is compounded by the changes in network topology. Our objective is to study the conditions which allow for packet delivery performance to adapt to the level of cooperation. We consider a basic information dissemination problem first, based on which we formulate the game structure and present analytical and simulation based insights on evolution of node cooperation behavior. In the next section, we use this formulation to study how node cooperation and information dissemination proceed in tandem in a more realistic setup.

\subsection{System Model}\label{subsec:sys-mod}
\subsubsection{Problem Definition}\label{subsubsec:pblm-defn}
The problem setup we consider is that of a network in which $M$ packets are to be disseminated. The packets are generated by individual source nodes, all of whom are cooperators initially. All nodes in the network are interested in receiving all of these $M$ packets. However, since a node is only interested in receiving, it may not choose to forward them to other nodes, thereby adopting a free-riding behavior. As successful dissemination of the packets relies on the number of nodes willing to forward, a tradeoff exists between node behavior and packet delivery performance.


\subsubsection{Packet Transmission Algorithm}\label{subsubsec:pkttx-algo}
The packet transmission algorithm proceeds as follows. Nodes buffer all packets they receive. Upon deciding to act as a cooperator, a node randomly chooses from the set of packets it has not yet transmitted and schedules it for transmission. If there are no new packets to be transmitted, any one of the packets is retransmitted. Allowing retransmission of packets is useful for dynamic scenarios such as in the case of mobility.

\subsubsection{Mobility Model}\label{subsubsec:mob-model}
We first use a model in which all nodes move with a constant velocity $v$ in randomly chosen directions. Thus, at any time $t$, the coordinates of a node $i$ is given relative to its position at time $(t-1)$ as,
\begin{align*}
x_i(t) &= x_i(t-1) + v \cos \theta_i \\
y_i(t) &= y_i(t-1) + v \sin \theta_i
\end{align*}
where the angle of movement $\theta_i$ is randomly chosen in each time slot.

\subsection{Formulation}\label{subsec:pgg-specifics}
While the above discussion outlines a framework for formulation of public goods games in wireless networks, we now look at individual aspects of the game formulation pertaining specifically to the problem of network wide dissemination of $M$ packets. We elaborate further on specific aspects of the game formulation. 

As outlined earlier, the benefits obtained by nodes in a PGG map to the number of packets received in the case of information dissemination. The game dynamics are determined by the payoffs received by a node in each time slot and the resulting strategy updates. Here, the behavior of information exchange in wireless networks differs somewhat from how benefits are obtained in a PGG and payoffs realized. The difference results from the fact that while the payoffs based on neighborhood contributions are instantly realized in a PGG, the same is not true with packet transmissions in wireless networks due to node contention behavior. A packet transmitted by a node can only be received correctly by its neighbors if no other nodes are transmitting simultaneously within the same neighborhood, as doing so would result in collisions. Thus, actual transmission of a packet is likely to take place later than the time when a node decides to enqueue a packet for transmission by acting as a cooperator. This leads to a dilemma with regard to how the benefits were realized. For instance, consider a node $A$ that chooses to act as a cooperator and enqueue a packet at a time $t$, which eventually gets transmitted at time $t+\tau$. In this intervening time period $\tau$, however, it is quite likely that the $A$ decides to switch its strategy based on neighborhood observations of payoffs. Thus, when the packet transmission actually takes place, $A$'s strategy is actually that of a defector or a free-rider. This could give rise to inconsistencies in the game dynamics as neighbors of $A$ may decide on their strategies based on observation of $A$ as a defector, which is a false notion since the packet transmission actually took place because of $A$'s behavior as a cooperator. Further, there is the additional question of whether, as with the traditional definition of PGGs, payoffs should be updated during the time period $(t,t+\tau)$.

We take the above mentioned factors into consideration for formulating the PGG. One way of characterizing such a scenario would be to consider the expected delay between a node's decision to cooperate and actual transmission of the packet. Alternatively, as with the variants of PGG described earlier in \ref{sec:PGG}, a node can be expected to make contributions in each time slot it is a cooperator, with the contribution obtained as the probability of transmitting in each slot. The transmission probability can be obtained in terms of the number of contending nodes, which, in a saturated scenario, is the entire node neighborhood, and thus correlates directly to the contributions from cooperators in the second variant of PGG. Considering fair random scheduling of nodes, the cost incurred by a cooperator $i$ in a single time slot can be given as,
\begin{equation}\label{eq:coop-cost}
\eta_i = \frac{c}{{k_i} + 1}
\end{equation}
where $k_i$ is the degree of node $i$. As $i$ has equal chances of transmitting in each of the $\tau_i=({k_i}+1)$ time slots, the total contribution made by a node is proportional to how long it stays a cooperator. Thus, if it stays a cooperator throughout the $\tau_i$ slots, the contribution equals $1$. Therefore, the total benefit received by a node in a single time slot is the aggregation of contributions from all cooperators in its neighborhood, and can be expressed as,
\begin{equation}
b_i = \frac{r}{{k_i} + 1} \sum\limits_{j \in \mathcal{N}_i \cup i} \frac{c}{{k_j} + 1} {s_j} \nonumber
\end{equation}
where $\mathcal{N}_i$ denotes the set of neighbors of $i$ and $s_j = 1$ if $j$ is cooperator and $0$ otherwise. A normalization factor $\frac{r}{{k_i} + 1}$ is introduced to take into account the impact of contention of the neighborhood of $i$, in addition to the contention experienced by each node $j$. Subsequently, the payoff received by a node $i$ depends on whether it is a cooperator or a defector,
\begin{align}\label{eq:pggi-pyoffs}
\pi^D_i &= b_i \nonumber \\
\pi^C_i &= b_i - \eta_i
\end{align}
An important difference of the above formulation with the traditional definition of PGG is that the payoff of a node is only the result of the game centered on itself. Such a design is due to the fact that nodes only receive packets transmitted in their immediate neighborhood. 

\begin{figure}
\centering
\includegraphics[width=.95\columnwidth]{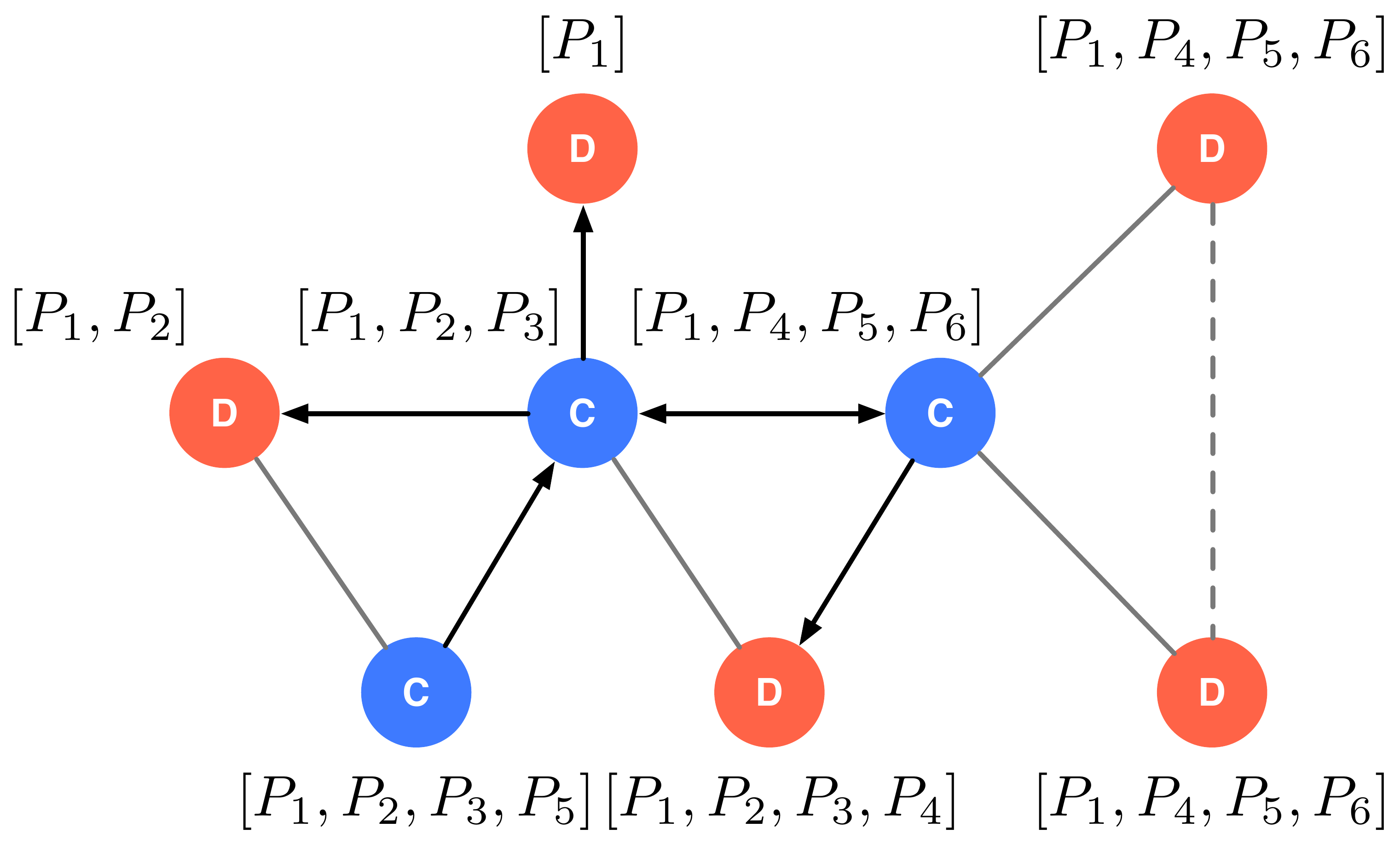}
\caption{Figure showing direction of flow of benefits from cooperators depending on the reception status of nodes in a graph. Darkly shaded lines with arrows indicate edges with flow of benefits. Lightly shaded lines indicate edges with no flow of benefits. Lightly shaded dashed lines show edges among defectors.}
\label{fig:pyfdir-grp}
\end{figure}

Next, we make the observation that, while the above formulation correlates directly with node contributions in a saturated scenario, the same is unlikely to be true for the packet dissemination problem under investigation since the set of contending nodes only includes cooperators. Thus, equation (\ref{eq:coop-cost}) can be rewritten as,
\begin{equation}\label{eq:mod-cost}
\eta_i = \frac{c}{{n^C_i} + 1}
\end{equation}
where $n^C_i$ denotes the number of cooperators in the neighborhood of $i$. As outlined earlier, the benefits in PGG correspond to the reception of packets. In a PGG centered at node $i$, the total benefit received is equal to the total of the contributions made by each cooperator multiplied by a factor $r$. In the case of the information dissemination, however, a packet received by a node is only meaningful if it has not been received earlier. Thus, the only cooperators in a node's neighborhood that matter are those that have at least one new packet to transmit. Based on this, we rewrite the benefit received by a node $i$ as follows
\begin{equation}\label{eq:nd-benefit}
b_i = \frac{r}{{n^C_i} + 1} \sum\limits_{j \in \mathcal{N}_i \cup i} \frac{c}{{n^C_j} + 1} {t_j}
\end{equation}
where $t_j = 1$ if $j$ is a cooperator and has at least one packet to transmit which has not yet been received by $i$ and $0$ otherwise. The neighborhood of $i$ is denoted by $\mathcal{N}_i$. Subsequently, the payoffs are expressed as in equation (\ref{eq:pggi-pyoffs}).

The above game formulation allows us to map the game behavior directly to that of information dissemination. The number of packets received by a node, thus, correlate to the total benefits accumulated over time by a node. 

The payoffs defined in equation (\ref{eq:pggi-pyoffs}) refer to the instantaneous payoffs corresponding to the reception performance in a single time slot. The instantaneous payoff values are used for strategy update by nodes.

\subsection{Conditions for Node Strategy Evolution}\label{subsec:condn-evoln}
To gain an understanding of how node behavior evolves, we obtain conditions under which the game formulation in section \ref{subsec:pgg-specifics} can lead to growth of cooperation in a network. To do so, we look at the neighborhood conditions that induce a node to change its behavior from that of a defector to a cooperator.

\subsubsection{Conditions for Inducing Cooperation in Static Network}\label{subsec:condn-coop}
Consider a node $i$ that acts as a defector at time $t$. The behavior of $i$ at time $(t+1)$ is determined by the difference in payoffs of itself and its neighbors. Let $k_i$ be the degree of $i$, $x_i$ the fraction of neighbors  which are cooperators and $u_i$ the fraction of cooperators that are useful to $i$. As defined earlier, $u_i$ is determined by the set of packets already received by $i$ at time $t$. The payoff of $i$ at time $t$ can, thus, be written as,
\begin{equation}\label{eq:pyfi-def}
\pi^D_i = \frac{1}{(x_i k_i + 1)}\frac{u_i x_i k_i}{(x(t) \langle k \rangle + 1)}
\end{equation}
where $x(t)$ is the fraction of cooperators at time $t$ in the network and $\langle k \rangle$ is the average node degree. Transition of $i$'s strategy to that of a cooperator at time $(t+1)$ is proportional to the difference in payoff $\pi^C_j$ of a cooperating neighbor $j$ with $\pi^D_i$. As with $i$, the payoff of $j$ is determined by the set of cooperators in its neighborhood which have packets to transmit not yet received by $j$, and can be given as,
\begin{equation}\label{eq:pyfj-coop}
\pi^C_j = \frac{1}{(x_j k_j + 1)}\frac{u_j x_j k_j}{(x(t) \langle k \rangle + 1)} - \frac{1}{(x_j k_j + 1)}
\end{equation}
The probability that $i$ becomes a cooperator is proportional to the payoff difference, which can be expressed as,
\begin{eqnarray}\label{eq:pydiff}
p_i(D \rightarrow C) & \propto & \frac{1}{(x(t) \langle k \rangle + 1)} \left( \frac{u_j x_j k_j}{(x_j k_j + 1)} - \frac{u_i x_i k_i}{(x_i k_i + 1)} \right) \nonumber \\
& &  - \frac{1}{(x_j k_j + 1)}
\end{eqnarray}
The above discussion shows that $j$ is more likely to induce cooperative behavior if it is likely to receive new packets from cooperators in its own neighborhood. Cooperative behavior is, thus, likely to grow from regions of high concentration of cooperators coupled with high diversity of source packets. Further, a high value of the synergy factor $r$ mitigates the role of the cost incurred by a cooperating node.

While the above discussion outlines the conditions that determine whether node $i$ changes its strategy, the probability $q_{ji}(t)$ that a neighbor $j$ contributes to $i$'s payoff at time $t$ depends on the set of packets already received by each node. Let $m_i$ and $m_j$ denote the number of packets already received by $i$ and $j$ respectively at time $t$. The probability $q_{ji}(t)$ is determined by the following lemma:

\begin{lemma}\label{lem:clust-imp}
The probability that a node $j$ contributes to the benefits received by a neighbor $i$ is given as

\begin{eqnarray} \label{eq:qji-clust}
q_{ji}(t) = \left\{ \begin{array}{ll} x(t) \mbox{,} &\mbox{if
$m_j > m_i$} \\
x(t) \left[1 - \left(\frac{c_j (k_j - 1)}{k_j}\right)^{m_j}\right] \mbox{,} &\mbox{otherwise}.
\end{array} \right.
\end{eqnarray}
where $c_j$ is the clustering coefficient of $j$, $k_i$ and $k_j$ are the node degrees of $i$ and $j$ respectively. 
\end{lemma}

\begin{proof}
\begin{enumerate}[(a)]
\item The first part of the result is obtained trivially as the condition $m_j > m_i$ implies that $j$ has more packets than $i$ and the probability that it contributes to $i$'s payoff is that it is a cooperator.
\item When $m_j \leq m_i$, $q_{ji}(t)$ is the probability that $j$ receives at least one packet not received by $i$ in the time period $[0,t]$. To understand this, we make observations about how packet dissemination proceeds in a static network. As the neighborhood of a node does not change with time, any broadcasted packet can only be received by a node if it is forwarded by any one of its neighbors. Two neighboring nodes $i$ and $j$ can, therefore, receive the same packet if it is broadcasted by a node which is a neighbor to both. Alternatively, $j$ can receive a packet from a node which is not a neighbor to $i$ and therefore, $i$ does not receive it. \footnote{There is a third possibility here, which is that $i$ receives the same packet from another neighbor $a$ which is not a neighbor of $j$. However, this implies that the packet reached both $j$ and $a$ without reaching $i$. The probability of this scenario is quite low \cite[Fig. 1]{Banerjee2012} and hence we do not include it as part of this discussion.} Thus, the probability that both $i$ and $j$ receive the same packet from a node $b$ is the probability that $b$ is a neighbor to both, which can be given in terms of the clustering coefficient as
\begin{equation}
p^{common}_{ij} = c_j \frac{(k_j - 1)}{k_j} . \nonumber
\end{equation}
Thus, the probability that all packets received by $j$ are also received by $i$ is $(p^{common}_{ij})^{m_j}$. The probability that $j$ contributes to $i$'s payoff is that it has received at least one packet which has not yet been received by $i$ and thus, equation (\ref{eq:qji-clust}) follows.
\end{enumerate}
\end{proof}

To further understand the implication of Lemma \ref{lem:clust-imp} on node cooperation behavior, we make the observation that node interactions in this context are comprised of two aspects, namely, (a) the influence a node has on the payoff of its neighbors, (b) the influence of a node on the changes in strategy of its neighbors. Lemma \ref{lem:clust-imp} directly shows how observation (a) is impacted by the neighborhood structure of a node, measured as the clustering coefficient. However, a further inference can be drawn from Lemma \ref{lem:clust-imp} in relation to observation (b). A high clustering coefficient $c_j$ implies that $j$ shares a greater part of its payoff with its neighbors, implying low difference in payoff values. Thus, the probability $p_i(D \rightarrow C)$ in equation (\ref{eq:pydiff}), for a defecting neighbor $i$ of $j$, to switch to a cooperating behavior by comparing its payoff to $j$ redcues with increasing $c_j$. An extreme case, therefore, is when $\mathcal{N}_j \subset \mathcal{N}_i$, $\mathcal{N}_i$ and $\mathcal{N}_j$ being the set of neighbors of $i$ and $j$, in which case $j$ can never influence $i$.

\subsubsection{Conditions for Cooperation Growth with Mobility}\label{subsec:anal-mob}
In a network with mobile nodes, node cooperation behavior is impacted due to the constant churning of node neighborhoods. As discussed earlier in section \ref{subsec:condn-coop}, the evolution of node cooperation behavior in a network is determined by how the node interactions impact (a) instantaneous payoffs, (b) strategy updates.

In section \ref{subsec:condn-coop}, we noted how the clustering coefficient of a node determines node behavior in its neighborhood. In a mobile scenario, however, the clustering coefficient cannot be directly correlated as it keeps changing due to constant changes in the neighborhood of a node. Since the payoff of a node is determined by the number of useful cooperators in its vicinity, it can receive higher payoffs if it meets cooperators with new packets. Existing research has shown \cite{Grossglauser2002,Ioannidis2009} that mobility can help improve packet delivery performance by reducing the number of hops required. For the problem under consideration too, mobility can impact node cooperation behavior due to changes in the neighborhood structure of nodes. Meeting a higher number of new nodes in each time step, which are likely to have new packets not yet received by a node, is likely to result in higher payoffs compared to maintaining the same set of neighbors. We analyze how mobility impacts the set of neighbors of a node in each time step.

For the mobility model used above, the jump length for a node $i$ at any time step $t$ is a constant distance $v$. The set of neighbors of $i$ as a result of this jump at time $(t+1)$ is determined by the jumps made by other nodes in the network, also of length $v$. The payoff of $i$ at $(t+1)$ increases if a greater fraction of its neighborhood consists of nodes which were not neighbors at $t$.

The probability that any neighbor $j$ of $i$ at time $t$ is also a neighbor at $(t+1)$ is the probability that $i$ and $j$ are located close enough to communicate with each other, i.e. they are located within each other's transmission range $rad$. Considering $N$ nodes distributed uniformly over a circular region of radius $R$, the probability that a node $i$ meets an existing neighbor, $p^{nbr}_{old}$, and the probability that it meets a new node, $p^{nbr}_{new}$ are respectively obtained as,
{\footnotesize 
\begin{align}\label{eq:pmeet}
p^{nbr}_{old} &= \int_{0}^{rad} \left(\int_{0}^{2 \pi} \frac{cos^{-1}(\frac{v^2 + {rad}^2 + d^2}{2 v d})}{\pi} \frac{1}{2 \pi} d\theta \right) \frac{2 x}{{rad}^2} dx \nonumber \\
p^{nbr}_{new} &= \int_{rad}^{R} \left(\int_{0}^{2 \pi} \frac{cos^{-1}(\frac{v^2 + {rad}^2 + d^2}{2 v d})}{\pi} \frac{1}{2 \pi} d\theta \right) \frac{2 x}{(R^2 - {rad}^2)} dx
\end{align}}

where $d = \sqrt{v^2 + x^2 - 2 x v \cos \theta}$. 
\begin{figure}
\centering
\includegraphics[width=.95\columnwidth]{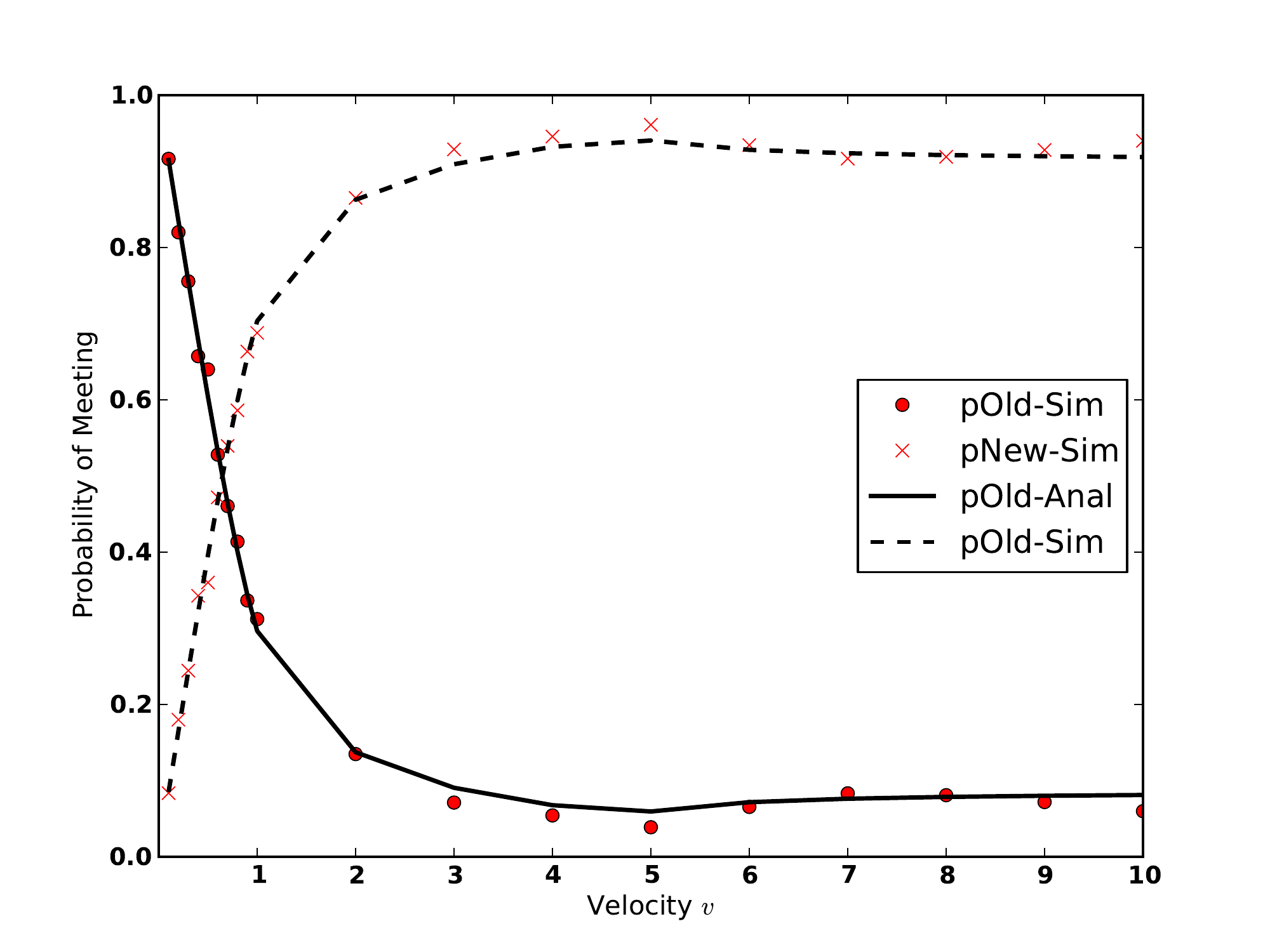}
\caption{Probability of meeting a new neighbor after making a jump.}
\label{fig:pnbrnew-constv}
\end{figure}
Subsequently, the probabilities that, among the set of neighbors of $i$ at $(t+1)$, a randomly chosen node is an old neighbor or a new one can be obtained respectively as,
\begin{align}\label{eq:pnewnbr}
f_{old} &= \frac{p^{nbr}_{old} {rad}^2}{p^{nbr}_{old} {rad}^2 + p^{nbr}_{new} (R^2 - {rad}^2)} \nonumber \\
f_{new} &= \frac{p^{nbr}_{new} (R^2 - {rad}^2)}{p^{nbr}_{old} {rad}^2 + p^{nbr}_{new} (R^2 - {rad}^2)}
\end{align}
A higher fraction of new neighbors, $f_{new} > f_{old}$, implies a higher number of useful cooperators in the neighborhood of a node resulting in a higher payoff. Any such node, thus, is more likely to influence strategies in its neighborhood. The variation of $f_{new}$ and $f_{old}$ with increase in velocity is shown in Fig. \ref{fig:pnbrnew-constv}. As can be seen, a higher node velocity increases the likelihood of meeting new nodes, which in turn can lead to higher payoffs as discussed earlier.


\begin{figure}
\centering
\includegraphics[width=.95\columnwidth]{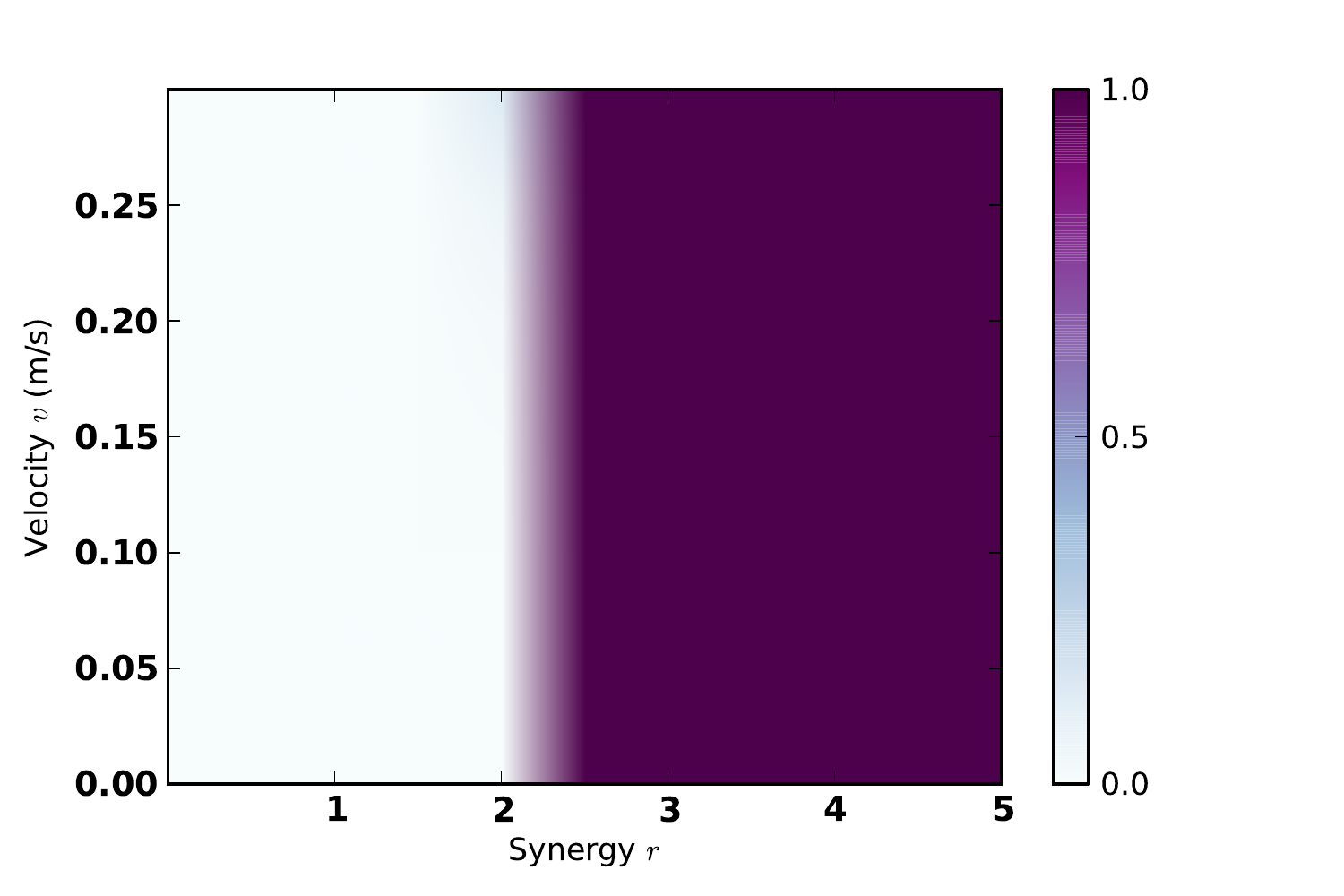}
\caption{Impact of synergy and velocity on node cooperation.}
\label{fig:fc-pgginfodissm}
\end{figure}

\subsection{Simulation}\label{subsec:pgginfodm-sim}

We use simulation results to understand the impact of node cooperation behavior on the information dissemination performance. We first run simulations on a uniformly distributed spatial network of $300$ nodes and observe the evolution of node behavior and algorithm performance over time averaged over $50$ simulation runs.

Each simulation run starts by randomly choosing a set of nodes as source nodes, each of which generates a single packet to be broadcasted to all nodes in the network. Initially, only the source nodes are cooperators. Subsequently, at each time instant, nodes determine their strategies based on payoffs obtained as per the game formulation in \ref{subsec:pgg-specifics}. Each time slot is assumed to consist of a single packet transmission. Nodes attempt to forward buffered packets only when they choose cooperation as a strategy. The payoffs are calculated as described in eq. \ref{eq:pggi-pyoffs}. The simulation is run till either all nodes use the same strategy or the information dissemination process is completed, i.e.\ all packets have been received by all nodes.

Fig. \ref{fig:fc-pgginfodissm} shows how node cooperation in the steady state is determined by values of $r$ and $v$. Compared to the basic PGG framework formulated in section \ref{sec:pggfwk-winet}, mobility is seen to enhance node cooperation significantly. This results from the fact that, in the PGG formulation adapted for information dissemination, node payoff is a function of the number of useful cooperators and not just the number of cooperators. As mobility speeds up the information dissemination process, nodes are more likely to encounter cooperators holding new packets and are thus 'useful', leading to higher payoffs and more chances of cooperative behavior.

\section{Node Cooperation Behavior for Content Downloading in Vehicular Networks}\label{sec:infdissm-vanet}
We now study node cooperation behavior in the context of content downloading services in a vehicular network. Content downloading services are likely to grow in popularity in future vehicular networks. Applications are likely to range from traffic related information such as safety and route information to user specific media content. Similar to existing studies \cite{Ioannidis2009,Malandrino2011,DiFelice2011}, we consider a scenario in which all vehicular nodes are interested in downloading content from a single service provider. The content dissemination is done by the service provider by distributing them to a set of seeder nodes. The seeder nodes then broadcast the packets to their immediate neighbors, which in turn choose to relay them to other nodes depending on their cooperation strategy. The node strategy evolution follows the same PGG formulation in section \ref{subsec:pgg-specifics}.

\subsection{Network Model}\label{subsec:cntdw-nwmod}
We consider a network of $N$ mobile nodes, all of whom are interested in receiving a set of $M$ packets from a content service provider. The nature of the content itself could depend on the type of vehicular application. For instance, a large media file could be composed of $M$ chunks. Alternatively, each packet could correspond to individual pieces of information in which all nodes are interested. We consider a scenario in which the service provider distributes all of these packets to a set of seeder nodes distributed randomly in the network. At each time slot, a seeder transmits a randomly chosen packet from the set of $M$ packets to its immediate neighbors. Dissemination to the rest of the network is determined by the relaying behavior of non-seeder nodes and therefore, depends on the evolution of node cooperation in the network. As in section \ref{sec:pgg-infodissm}, the strategy adopted by a node at each time slot depends on the payoff obtained by its immediate neighbors, determined by the number of useful cooperators in the neighborhood. 

A similar setup was considered in \cite{Li2012a}, where a small subset of nodes in the network receive content directly from the service provider. Packet forwarding is limited to a set of helper nodes, which includes nodes receiving content directly from the service provider, while the rest of the nodes in the network are marked as subscribers which only receive packets. However, the impact of evolution of node strategies is not studied in \cite{Li2012a}. We focus on how information dissemination is impacted when nodes are allowed to switch strategies depending on the payoffs received. A key difference with the scenario in section \ref{sec:pgg-infodissm} is the presence of a dedicated set of seeder nodes which always maintain their strategy as cooperators. Such a provision is necessary since seeder nodes always begin with the complete set of packets. Therefore, allowing them to change strategies using the PGG formulation would immediately result in them acting as defectors, thereby halting the dissemination process. Moreover, from the perspective of practical application, this corresponds to a scenario in which a service provider deploys seeder nodes to improve data delivery performance at receivers.

All nodes move according to the Levy Walk mobility model \cite{Rhee2008,Lee2012}. While a majority of existing research has illustrated the accuracy of Levy Walks in predicting human mobility, recent research has shown that the same is also true for vehicular mobility \cite{Li2012a,Li2012,Jiang2009}. Moreover the study of cities road networks in \cite{Barthelemy2011, Masucci2009} show that the typical distance $\ell_{1}$ between nodes intersection scale as $P(\ell_1) \sim \ell^{-\gamma}_{1}$ suggesting that Levy Walk will be a meaningful approximation (given the existence a cutoff) of the jump length distribution of cars from intersection to intersection in a roadmap. Hence, adopting such a model presents a realistic setup for studying node cooperation behavior in the presence of mobility (cf. Fig. \ref{levywalk}).

\begin{figure*}[tb]
    \centering
    \subfigure[]{\includegraphics[width=.95\columnwidth]{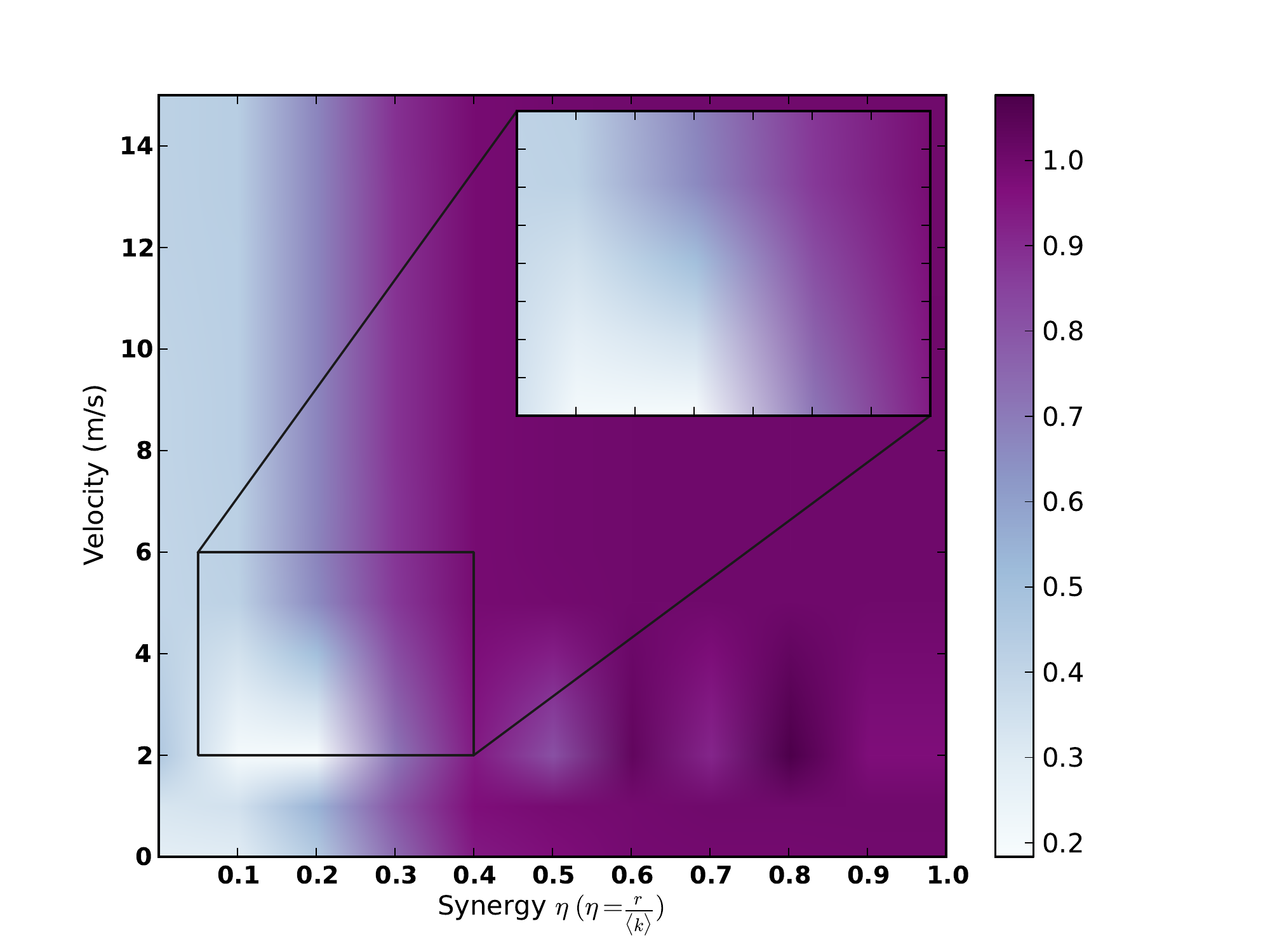} \label{subfig:cmCD30s}}
    \subfigure[]{\includegraphics[width=.95\columnwidth]{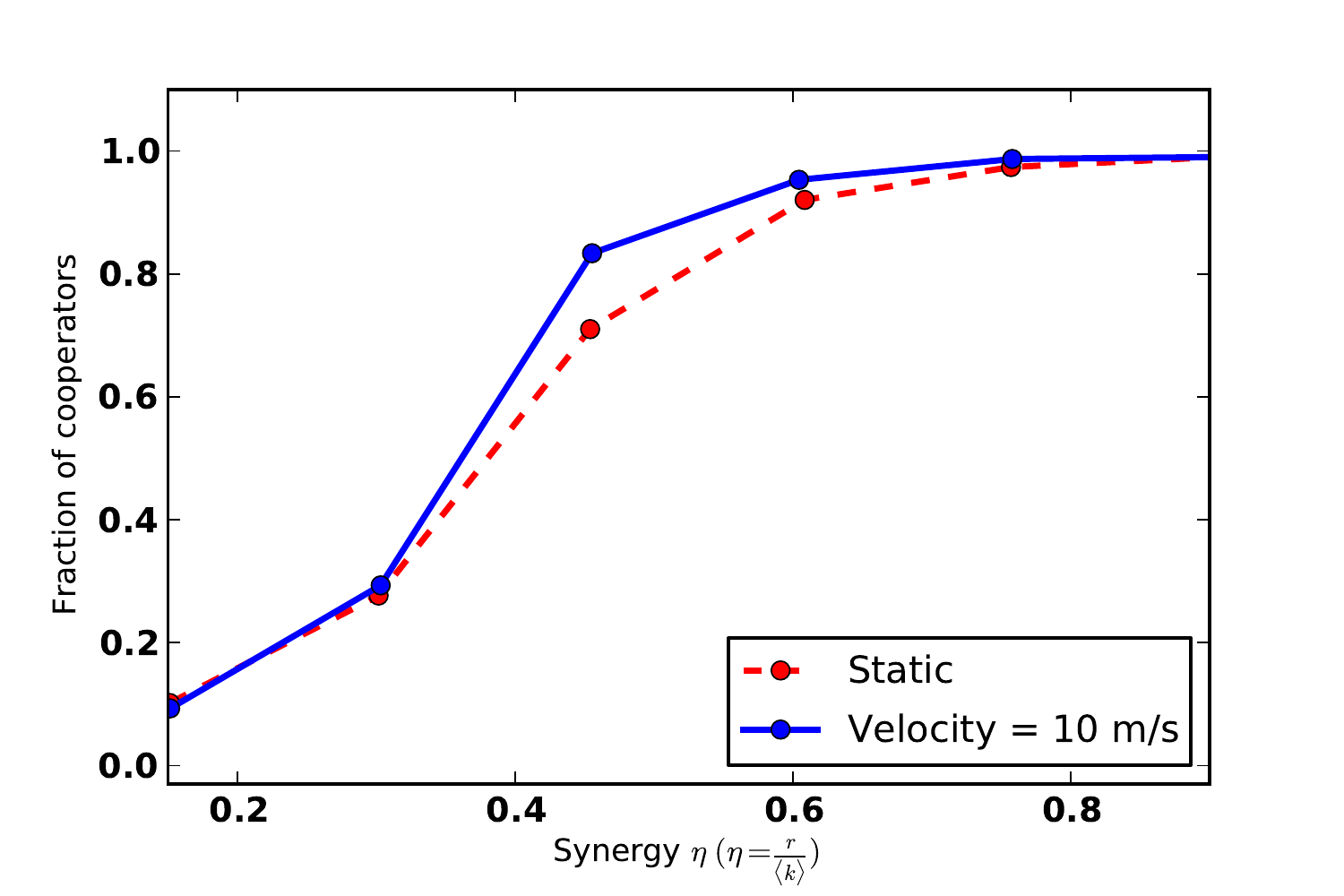} \label{subfig:cmCD30s2}}
    \subfigure[]{\includegraphics[width=.95\columnwidth]{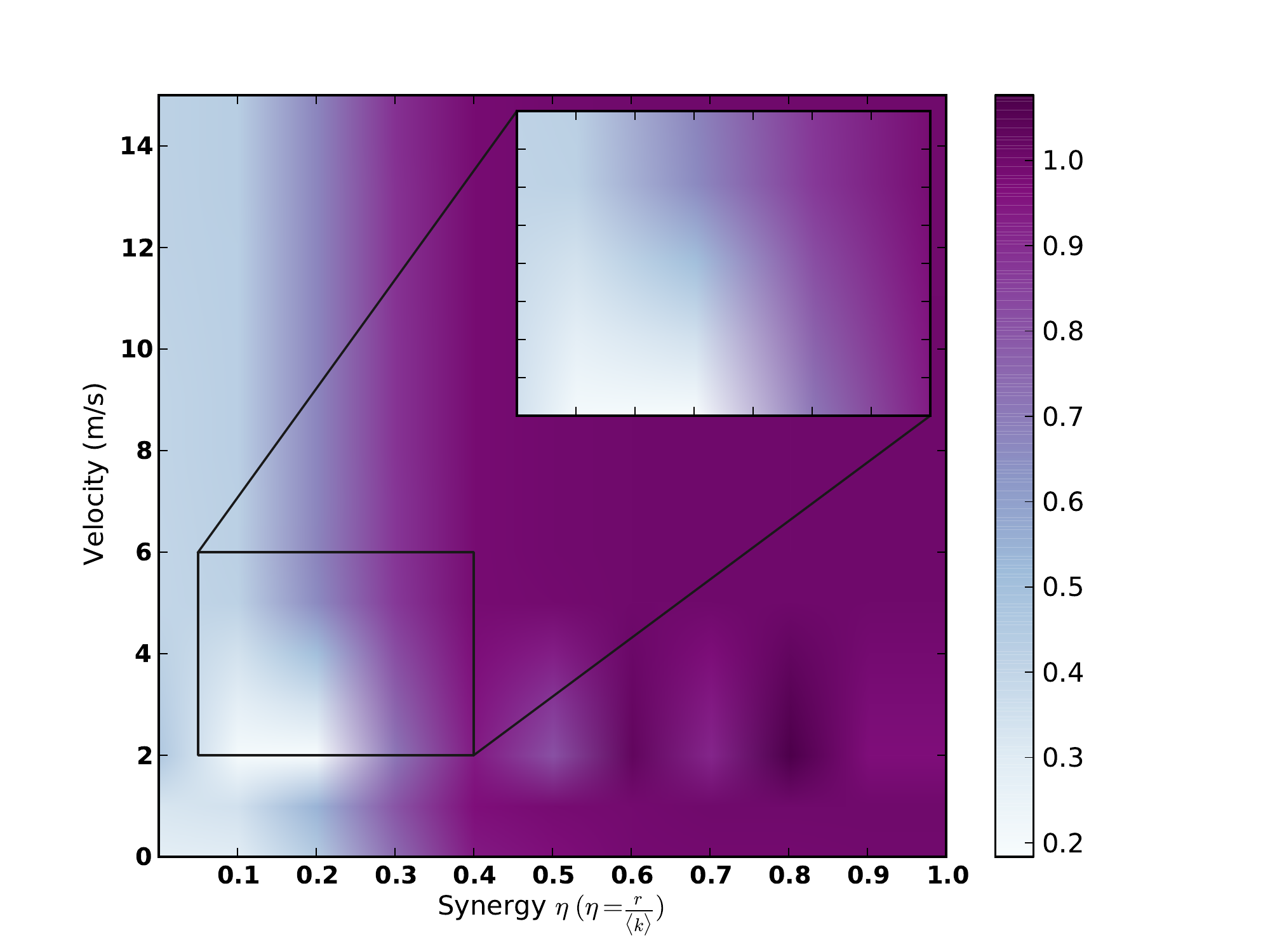} \label{subfig:cmCD60s}}
	\subfigure[]{\includegraphics[width=.95\columnwidth]{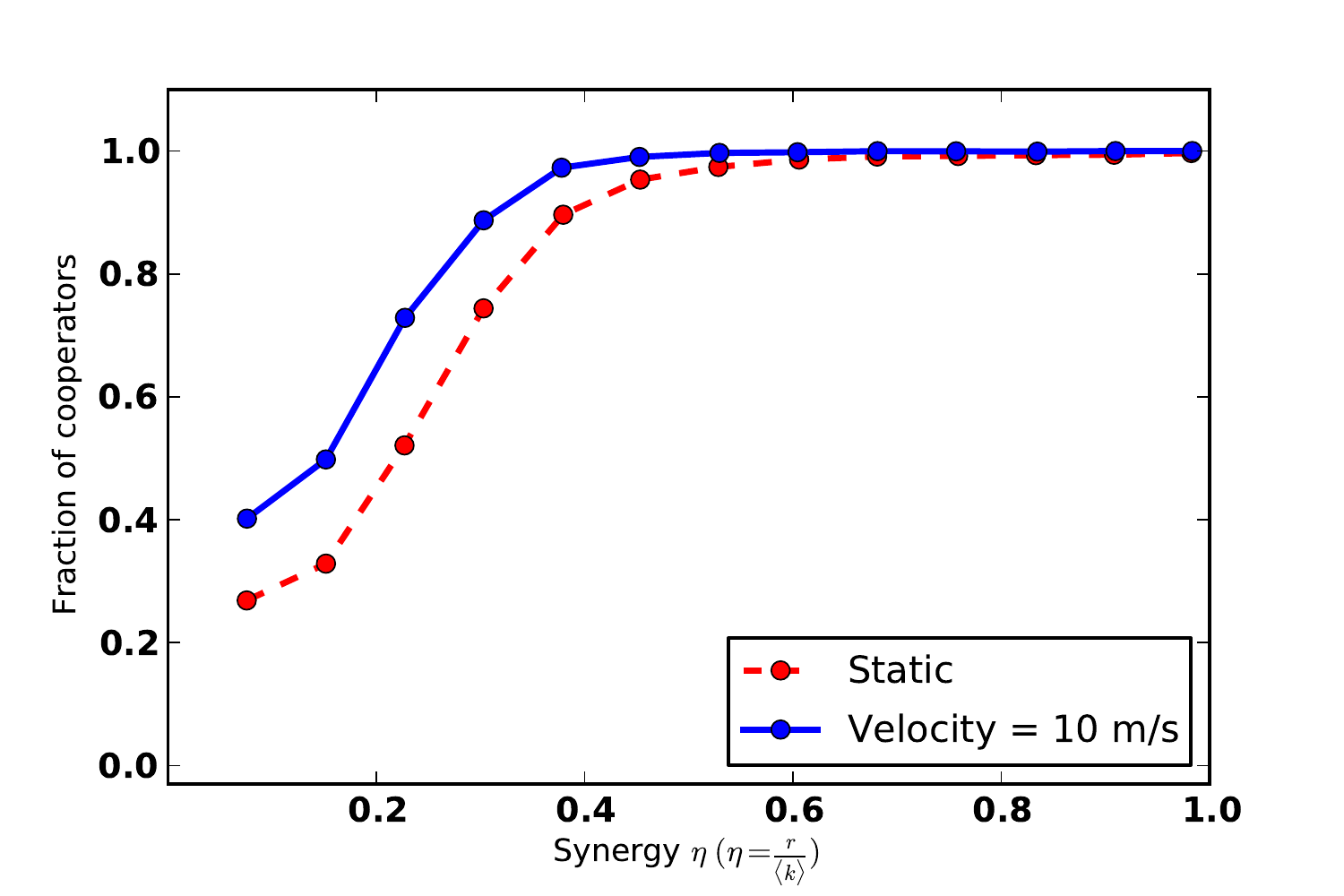} \label{subfig:cmCD60s2}}
    \caption{Node cooperation level in the steady state, (a) for the case of a networks with 30 seeders, (b) subset of the simulations with velocity $0\ m/s$ and $10\ ms/s$ for the case of 30 seeders, (c) full set of simulation for the case of a networks with 60 seeders,(d) subset of the simulations with velocity $0\ m/s$ and $10\ ms/s$ for the case of 60 seeders. The initial fraction of cooperator in Fig (b) and (d) account from the fact that the seeders are alway cooperator}
    \label{fig:cntdwld_ndcoop}
\end{figure*}

\begin{figure}[tb]
     \centering
     \includegraphics[width=.95\columnwidth]{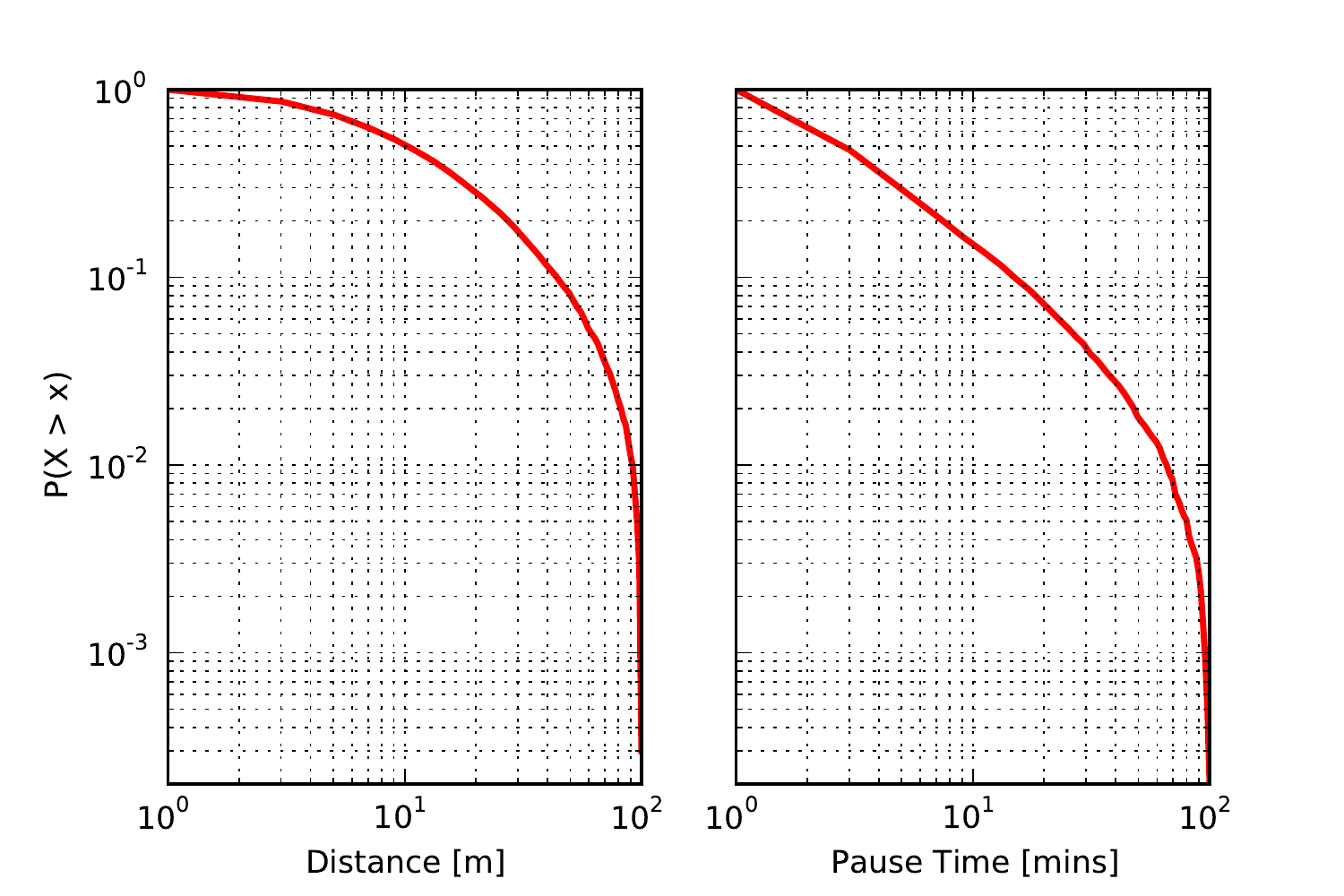}
     \caption{%
        Levy walk model used in our simulations, we plot an example with the velocity of node set to 5m/s, (cf. Table \ref{tab:parameters} for more information about the other parameters)
     }%
   \label{levywalk}
\end{figure}

\begin{figure*}[htb]
    \centering
    \subfigure[$60$ seeders]{
    \includegraphics[width=.8\textwidth,height=3in]{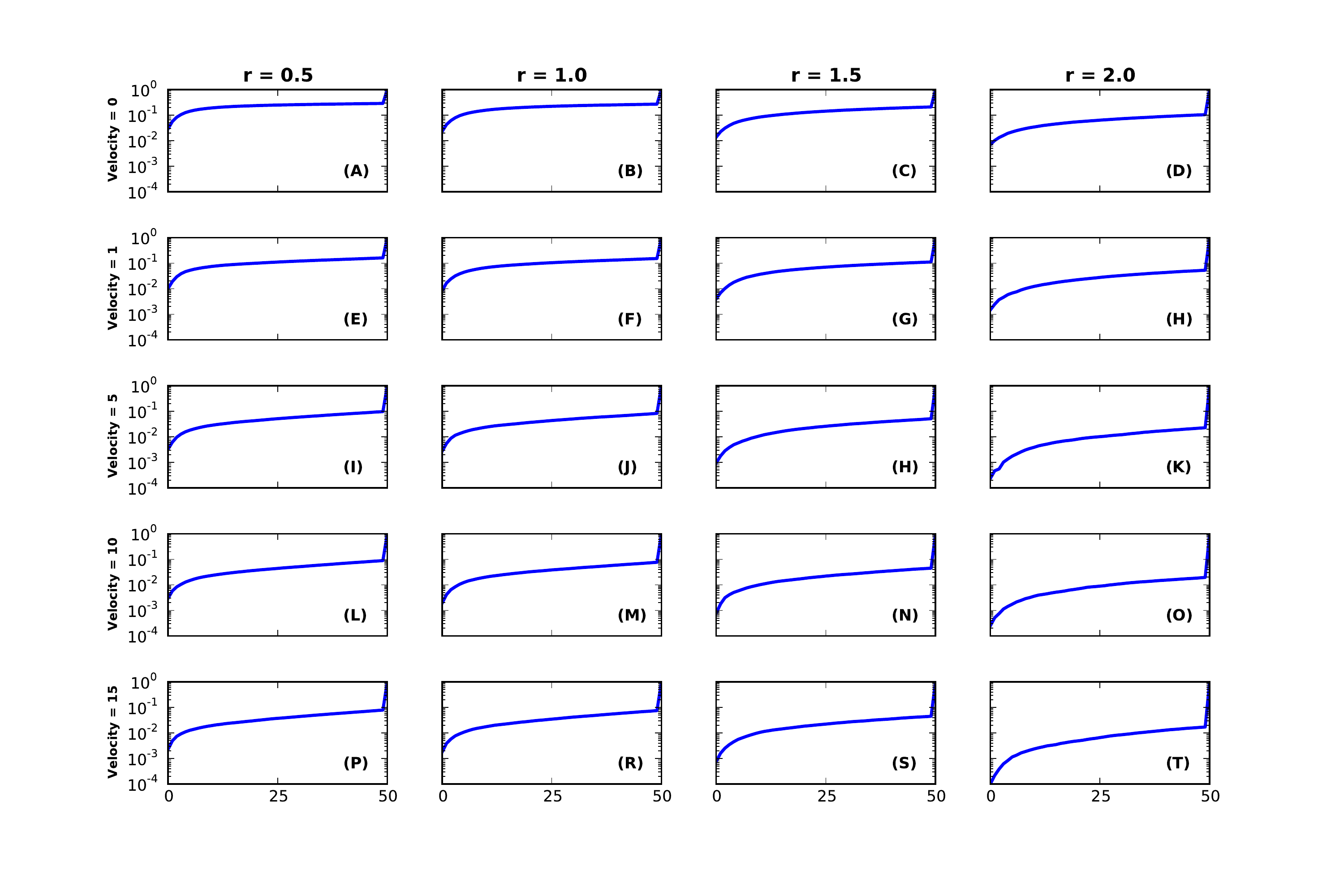}
    \label{subfig:nPD_30seed}
    }
    \subfigure[$30$ seeders]{
    \includegraphics[width=.8\textwidth,height=3in]{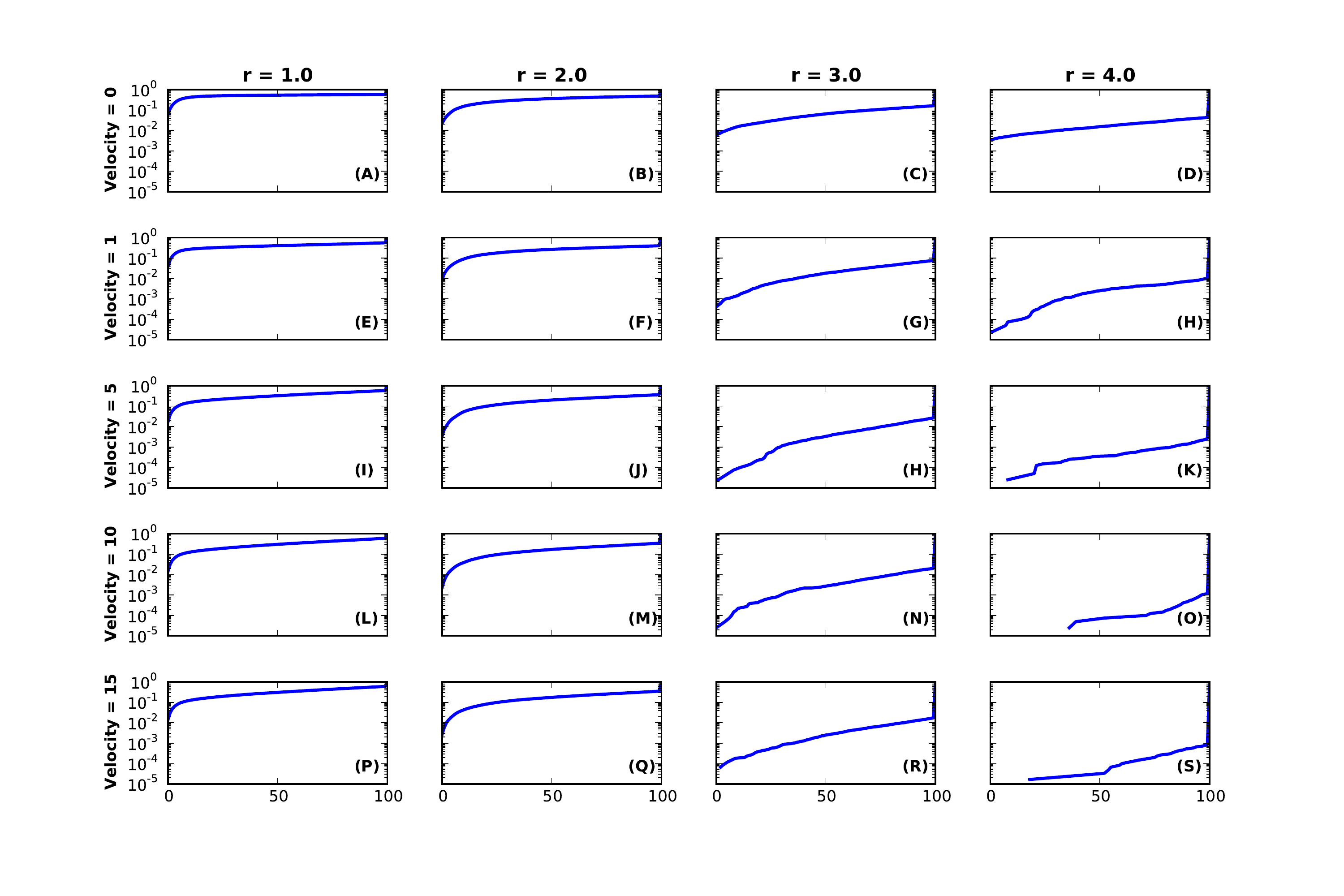}
    \label{subfig:nPD_60seed}
    }
    \caption{The figures (a) and (b) show the evolution of distribution of packets received by nodes across different velocity and synergy factor ($r$). Both figures showed that both factors greatly improve the spreading efficiency of the packets across the networks, from weak diffusion of packets in fig. (a).A to almost full diffusion in fig. (a).T (the same is true for fig (b).A and (b).S respectively). In fig	(a) the diffusion append for lower synergy factor than in fig (b) due to the fact there are twice the number of seeders. The buffer size for fig. (a) is 50 packets and for fig. (b) 100 packets. In fig (a) in the case of 60 seeders for $r > 2$ we achieved full diffusion in almost all the cases}
    \label{fig:pktDist_cntdwld}
\end{figure*}

\begin{figure}[tb]
     \centering
     \includegraphics[width=.95\columnwidth]{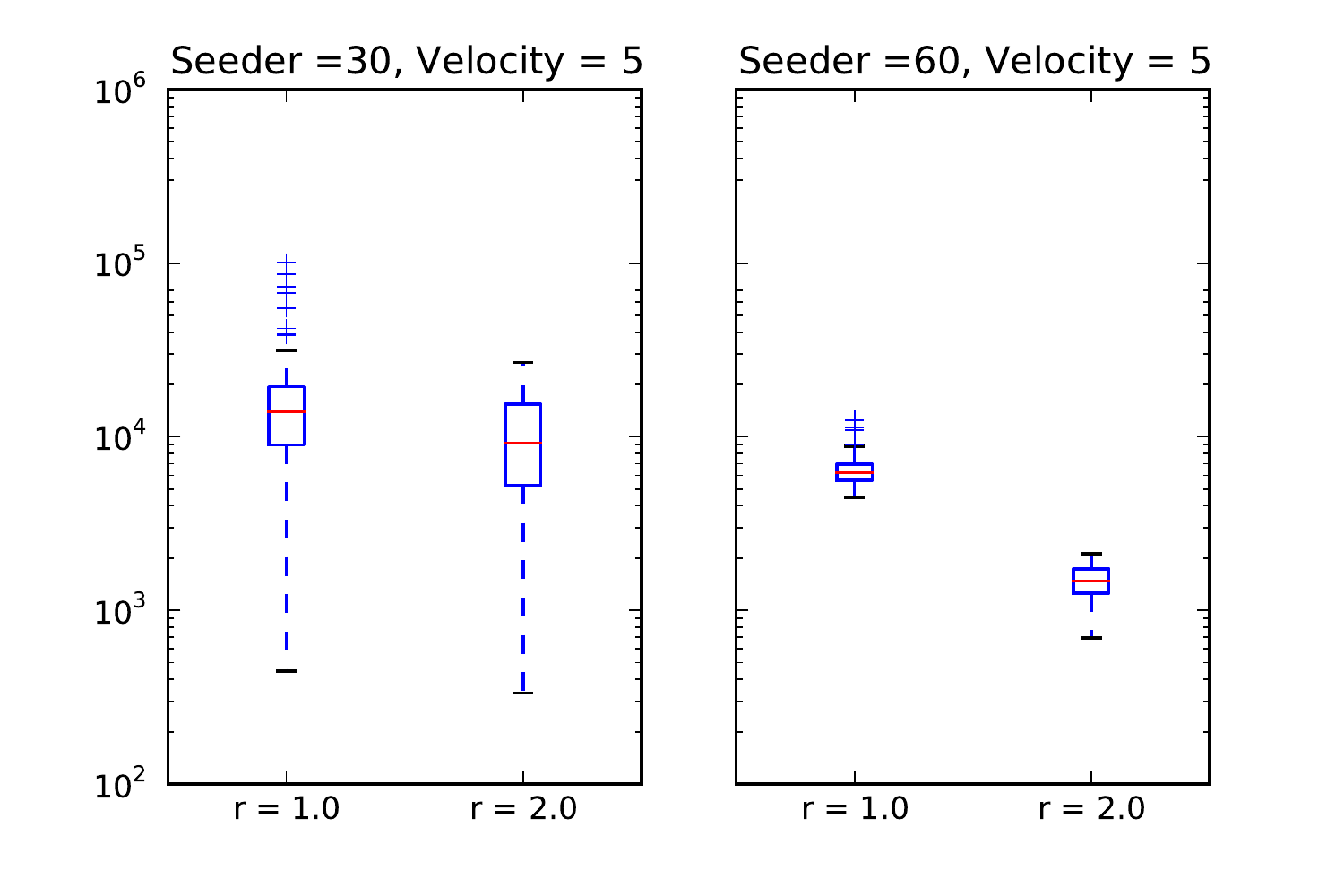}
     \caption{%
        Number of packets remaining to be delivered for different number of seeder. 
     }%
   \label{figquantile1}
\end{figure}

\begin{figure}[tb]
     \centering
     	\subfigure[]{%
            \label{fig:first}
            \includegraphics[width=.45\columnwidth]{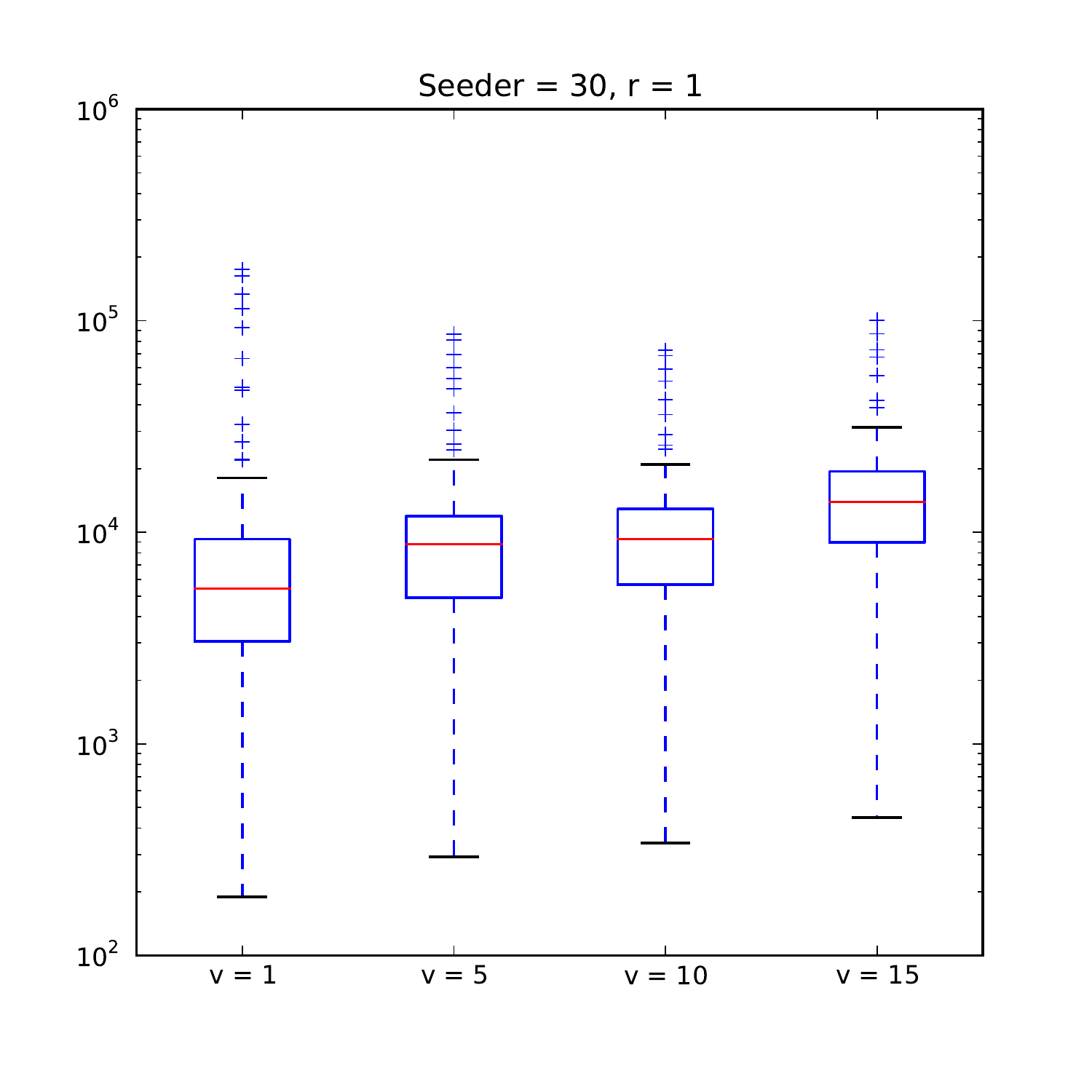}
        }%
        \subfigure[]{%
           	\label{fig:second}
           	\includegraphics[width=.45\columnwidth]{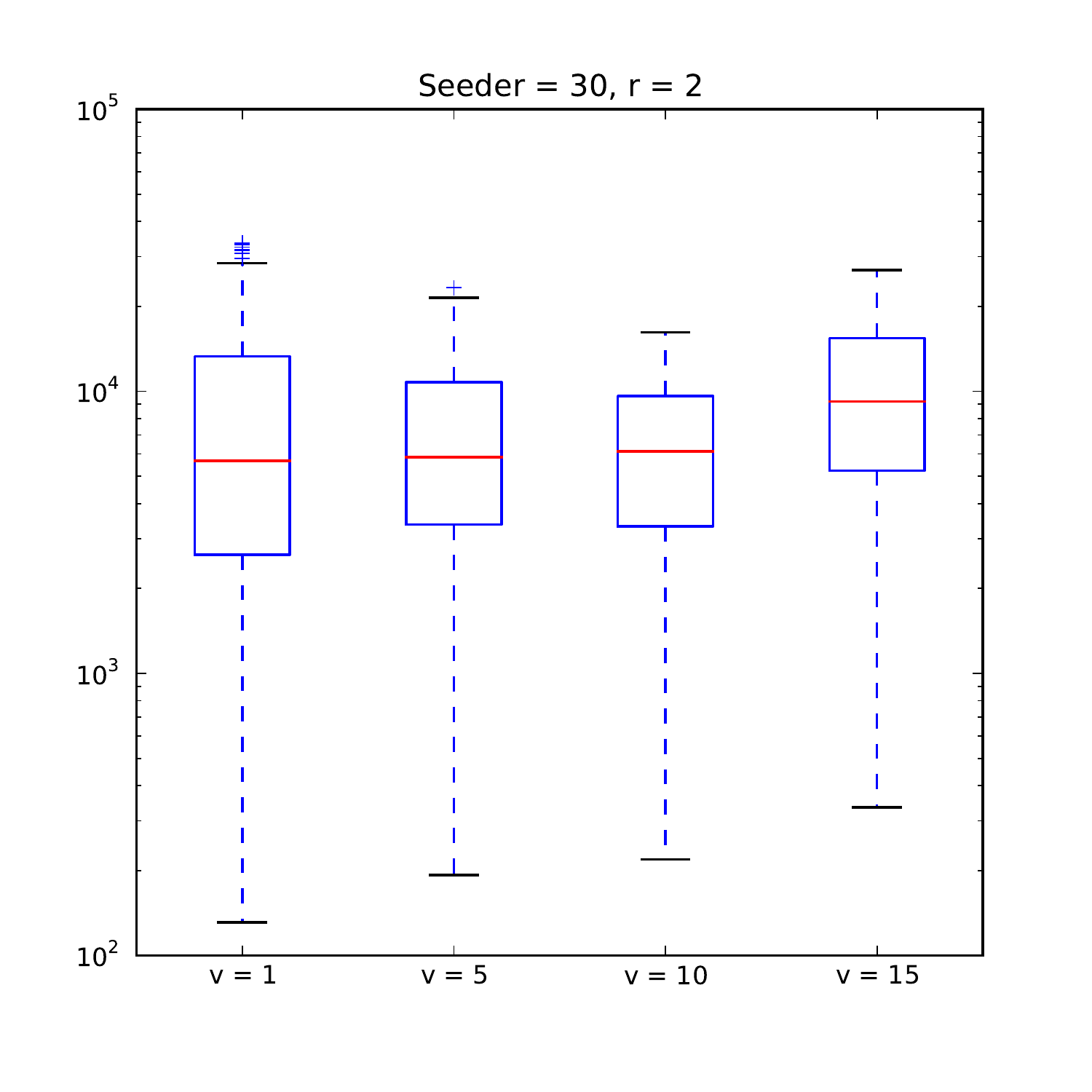}
        }
        \subfigure[]{%
			\label{fig:third}
            \includegraphics[width=.45\columnwidth]{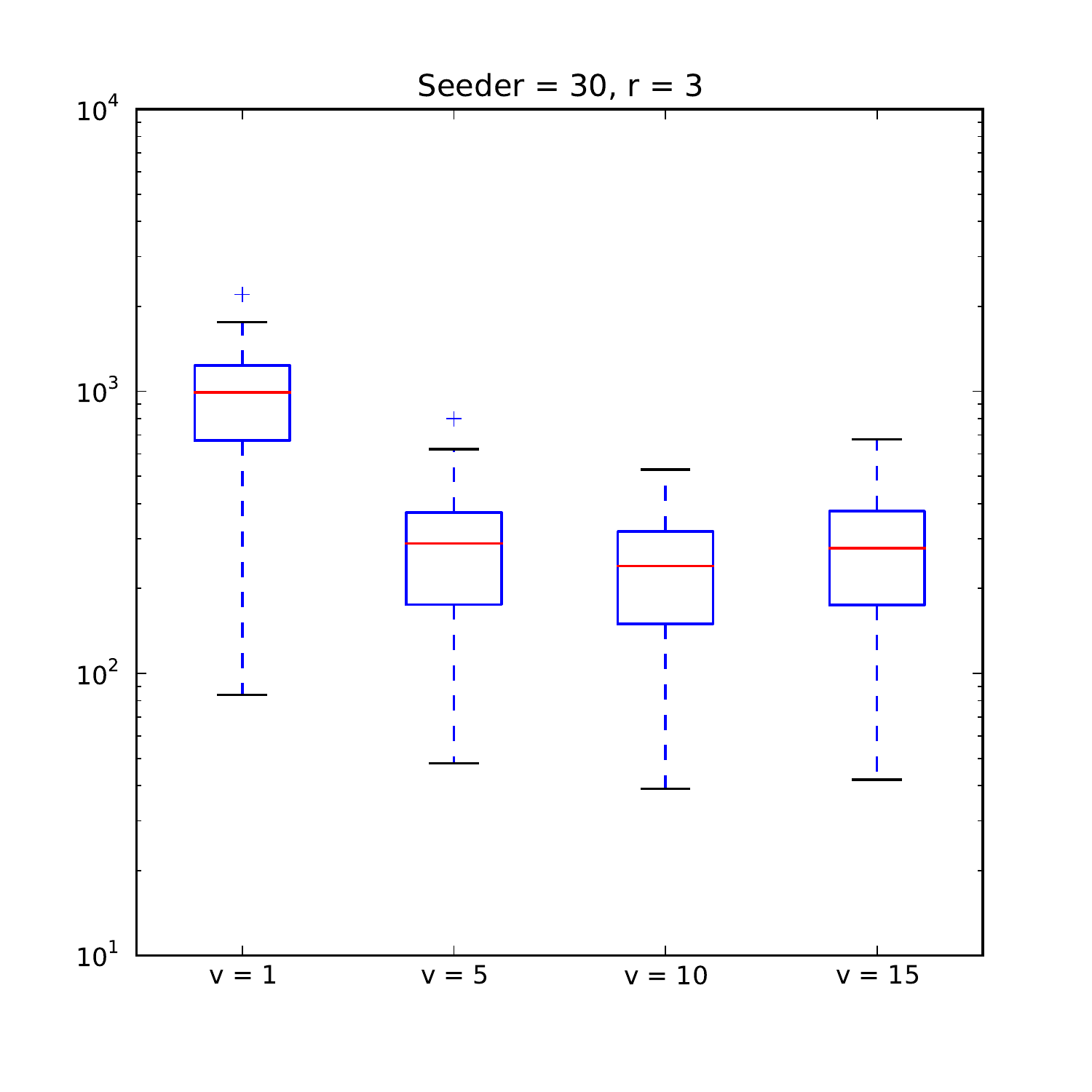}
        }%
        \subfigure[]{%
            \label{fig:fourth}
            \includegraphics[width=.45\columnwidth]{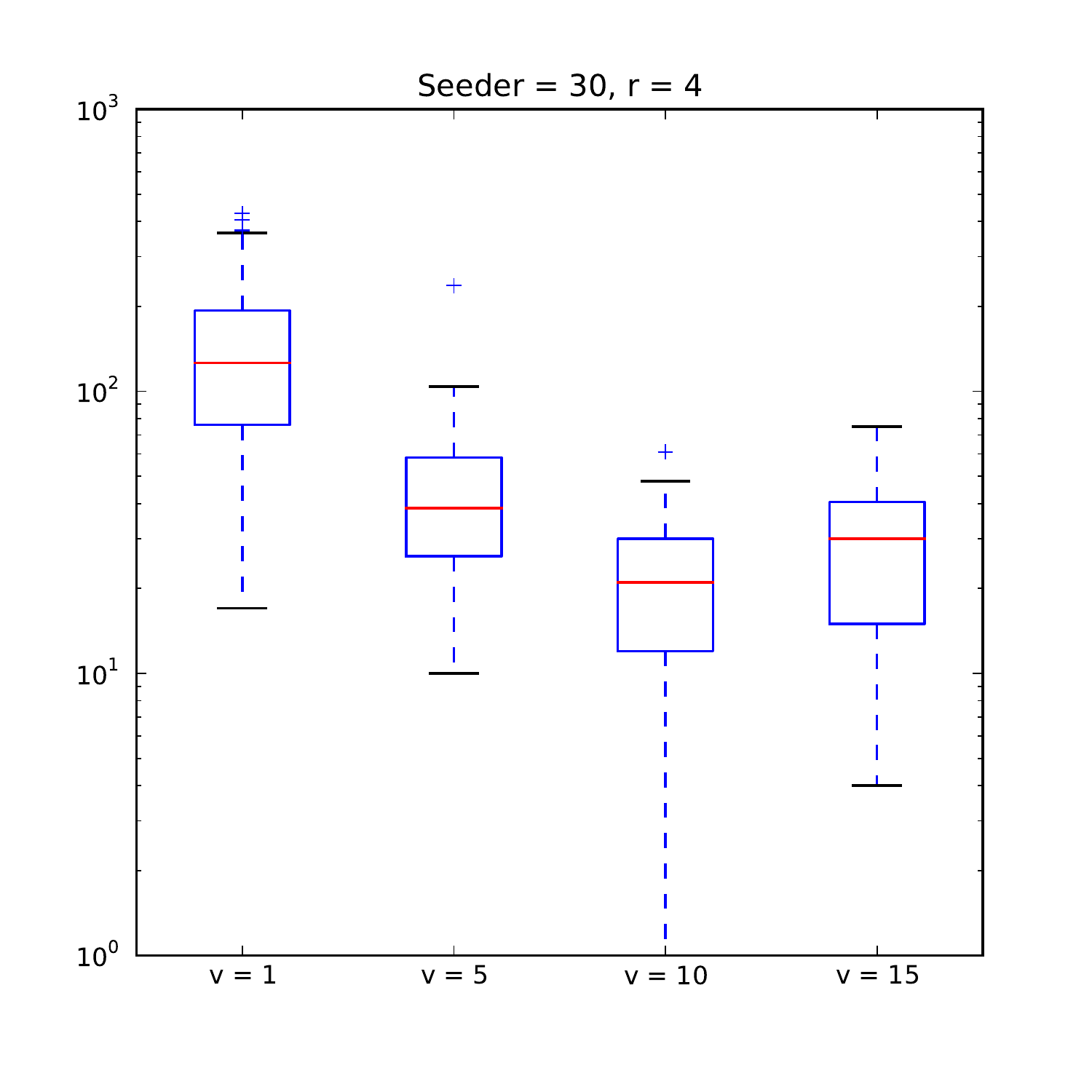}
        }%
    \caption{%
	Number of packets remaining to be delivered for various parameters $v$ and $r$ 
     }%
   \label{figquantile2}
\end{figure}

\subsection{Simulation Results}\label{subsec:simres-infodissm}
Our simulation setup consists of a network of $400$ nodes, a subset of which are seeders. We obtain results for different values of the number of seeders in the network.  We run simulations on a more realistic setup of randomly distributed nodes in which the neighborhood of each node is determined by its transmission range. We used the quasi unit disk model where $Pc$ the probability to connect to neighbor node is given relative to their distance $x$ (cf. eq. \ref{eq:pc}), 
{\footnotesize 
\begin{equation}\label{eq:pc}
   Pc(x) = 
   \left \{
   \begin{array}{ l c r}
   	 1 \ & \iff & x <  R_{inner} \\
     1-(\frac{R_{outer} - x}{R_{outer} - R_{inner}} )^\zeta \ & \iff & R_{outer} \geq   x  \geq R_{inner} \\
     0 \ & \iff & x  > R_{outer} \\
   \end{array}
   \right .
\end{equation}}

In order to allow for errors in calculation of payoffs in a realistic scenario, we include a gaussian noise parameter which introduces a certain degree of randomness,
\begin{equation}\label{eq:pggi-pyoffs-noise}
	\begin{array}{r c l}
		\pi_{i}^{D} & = & b_{i} + \mathcal{N}(0,\sigma)\\
		\pi_{i}^{C} & = & b_{i} + \eta_{i} + \mathcal{N}(0,\sigma)\\
	 \end{array}
\end{equation}
The parameter values used in the simulation are listed in Table \ref{tab:parameters}.   

%
%

\begin{table*}[!t]
\caption{Simulations Parameters}
\begin{tabular}{|c|c||c|c|}
\hline
Parameters & Value & Parameters & Value\\
\hline\hline
\multicolumn{2}{|c||}{Simulations Parameters} & \multicolumn{2}{c|}{Levy-walk Parameters}\\
\hline
Size of simulation area & 1000.0 $m^2$ & $\alpha$ & 0.9 \\
Number of node & 400 &  $\beta$ & 0.9 \\
Simulation length & 300.0 s & velocity & [0 .. 15.0] m/s \\
\cline{3-4}
Number of seeder & [30, 60] & \multicolumn{2}{c|}{Quasi-unit-disk Parameters} \\
\cline{3-4}
Buffer size & [50, 100] & $\zeta$ & 0.3 \\
\cline{1-2}
\multicolumn{2}{|c||}{PGG Parameters} & $R_{inner}$ & 40.0 m\\
\cline{1-2}
Initial cooperator ratio & 0.5 & $R_{outer}$ & 75.0 m\\
Noise variance ($\sigma$) & 0.1 & & \\
\hline
\end{tabular}
\centering
\label{tab:parameters}
\end{table*}

Fig. \ref{fig:cntdwld_ndcoop} shows the impact of the number of seeders on the node cooperation level in the steady state. Increasing the number of seeders results in a higher level of cooperation in the network. This is an expected result as more seeders implies faster dissemination of packets in the network, thereby leading to higher likelihood of finding 'useful' nodes. Thus, cooperating nodes have a higher probability of influencing neighborhood behavior. A surprising observation, however, is that, in the $\eta-v$ region in which cooperation dominates, the actual level of cooperation is higher when the number of seeders is lower. This results due to the fact that, while increasing the number of seeders results in a higher likelihood of cooperative behavior, this is accompanied by faster dissemination of information. Thus, nodes quickly receive the entire set of packets they are interested in. A small fraction of nodes are, thus, likely to stay non-cooperative. However the cooperation level remain only  mildly impacted by the  mobility of users (in the reasonable range simulated in this simulation). 

In Fig. \ref{fig:pktDist_cntdwld} we show the joint effect cooperation level and the dissemination process through the cumulative distribution of the number of packets received by nodes across the network. The ECDF shifts towards the right for higher values of $r$ and $v$ signifying improved dissemination of packets as nodes are more likely to get all packets. We show here that the mobility and the cooperation level ($r$ stands for the synergy factor) have both a great impact on the speed of dissemination process through the networks. The same phenomenon is observed in Fig. \ref{figquantile1}, and Fig.  \ref{figquantile2}, in both cases we shown how the set of packets pending to be delivered evolve as function of the velocity and the cooperation level. In Fig \ref{figquantile1} we show that the number of packet remaining to be delivered decrease as function of both R (the cooperation level), and the number of seeders. In Fig \ref{figquantile2} in the case of 30 seeders, we can  observe a sharp decrease of the number of packets pending to be delivered for different $r$ values, indeed in \ref{figquantile2}(a)  and \ref{figquantile2}(b) the number of packets  stay stable over different velocities but in \ref{figquantile2} (c) and \ref{figquantile2} (d) once a given value a the threshold value of $r$ is crossed the number of packet drops significantly and the information dissemination took place for almost all the nodes in the network.

\subsubsection{Simulation Code}
The code used in this paper is available at \cite{git}

\subsection{Comparison with Classical Game Theory based Approaches}\label{subsec:comp-cgt}
The motivation for the choice of an evolutionary game theory based approach was discussed earlier in section \ref{sec:intro}. To particularly highlight the benefits of our proposed model in the context of content downloading in vehicular networks, we draw comparison with approaches in existing literature that use classical game theory in similar deployments. In \cite{schwartz2011analysis}, a node computes the utility of each data file by taking into consideration the vehicle movement details such as the route and velocity and also characteristics of the data such as priority, size, age and geographical region. Two vehicles that encounter each other determine the set of packets to exchange using an algorithm based on Nash Bargaining. The same authors propose a similar approach in \cite{schwartz2012achieving} with the objective of achieving fairness in data dissemination. Shrestha et al. proposed a utility function based on packet priorities in \cite{shrestha2008wireless} and subsequently evaluated the performance of different bargaining algorithms.

A primary distinguishing feature of our proposed model with the above approaches is the localized nature of the information, both from the spatial and temporal perspectives. Not only is the information required limited to the immediate neighborhood, nodes are also assumed not to have any predetermined knowledge such as vehicle routes, packet priorities, etc. as all such information itself may be dynamic. Allowing myopic node decision making allows us to study how the network state determines its behavior in the succeeding time period. We envision this study to be useful from the point of view of planning for information dissemination in vehicular networks. For instance, in the scenario considered in section \ref{sec:infdissm-vanet}, seeder nodes may choose the set of packets to transmit depending on their immediate neighborhood, so as to maximize dissemination. Further, while the proposed model only takes into consideration packet reception status at nodes, it can be enhanced to accommodate other factors which can play a role in node decision making. For instance, factors such as residual energy may determine a node's forwarding behavior, which can be included in the model.

%
%
%
%
%
%
%
\section{Conclusion}
We introduce in this paper a framework based on Evolutionary Game Theory to study the evolution of node cooperation in wireless ad hoc networks. After introducing the general feature of the EGT, we proposed a model of node cooperation behavior that adapts Public Goods Games to the specifics of wireless ad hoc networks. Further on, we  described in which conditions networks could evolve and sustain a state of near to full cooperation among the devices. Our main contribution in this work is aimed at demonstrating the joint effects of the cooperation and the mobility rate on the spread of the information in wireless multi-hop networks. To fully take advantage of a rapid dissemination process in a rapidly changing environment we show that we need to exploit in tandem both the mobility and self-propagating property of the EGT to enable control of the node behavior. We also show how the success of the resulting diffusion process depends on the nature of the interactions between this two properties. Finally, we envision that the proposed model can be enhanced to develop a deeper understanding of cooperation in wireless networks by incorporating other parameters that impact node behavior such as energy consumption. Moreover, as the proposed model is centered around group interactions, the model can be used to study other forms of data delivery in wireless networks such as multi-hop unicast.

%
%
%
%
%
%
%

\bibliographystyle{IEEEtran}
\bibliography{pgg_winet_paper}

\begin{thebibliography}{10}
\providecommand{\url}[1]{#1}
\csname url@samestyle\endcsname
\providecommand{\newblock}{\relax}
\providecommand{\bibinfo}[2]{#2}
\providecommand{\BIBentrySTDinterwordspacing}{\spaceskip=0pt\relax}
\providecommand{\BIBentryALTinterwordstretchfactor}{4}
\providecommand{\BIBentryALTinterwordspacing}{\spaceskip=\fontdimen2\font plus
\BIBentryALTinterwordstretchfactor\fontdimen3\font minus
  \fontdimen4\font\relax}
\providecommand{\BIBforeignlanguage}[2]{{%
\expandafter\ifx\csname l@#1\endcsname\relax
\typeout{** WARNING: IEEEtran.bst: No hyphenation pattern has been}%
\typeout{** loaded for the language `#1'. Using the pattern for}%
\typeout{** the default language instead.}%
\else
\language=\csname l@#1\endcsname
\fi
#2}}
\providecommand{\BIBdecl}{\relax}
\BIBdecl

\bibitem{Chiasserini2003}
C.~Chiasserini and R.~Rao, ``{Cooperation in wireless ad hoc networks},'' in
  \emph{IEEE INFOCOM 2003. Twenty-second Annual Joint Conference of the IEEE
  Computer and Communications Societies (IEEE Cat. No.03CH37428)},
  vol.~2.\hskip 1em plus 0.5em minus 0.4em\relax IEEE, 2003, pp. 808--817.

\bibitem{Nowaka}
M.~Nowak, \emph{{Evolutionary Dynamics: Exploring the Equations of
  Life}}.\hskip 1em plus 0.5em minus 0.4em\relax The Belknap Press, 2006.

\bibitem{Weibull1997}
J.~W. Weibull, \emph{{Evolutionary Game Theory}}.\hskip 1em plus 0.5em minus
  0.4em\relax The MIT Press, 1997.

\bibitem{Nowak1994}
M.~A. Nowak and R.~M. May, ``{Superinfection and the evolution of parasite
  virulence.}'' \emph{Proceedings of the Royal Society B: Biological Sciences},
  vol. 255, no. 1342, pp. 81--89, 1994.

\bibitem{Friedman1991}
D.~Friedman, ``{Evolutionary Games in Economics},'' \emph{Econometrica},
  vol.~59, no.~3, pp. 637--666, 1991.

\bibitem{Smith1973}
J.~M. SMITH and G.~R. PRICE, ``{The Logic of Animal Conflict},'' \emph{Nature},
  vol. 246, no. 5427, pp. 15--18, Nov. 1973.

\bibitem{Nowak2002}
M.~A. Nowak, N.~L. Komarova, and P.~Niyogi, ``{Computational and evolutionary
  aspects of language.}'' \emph{Nature}, vol. 417, no. 6889, pp. 611--7, Jun.
  2002.

\bibitem{Lieberman2005}
E.~Lieberman, C.~Hauert, and M.~A. Nowak, ``{Evolutionary dynamics on
  graphs.}'' \emph{Nature}, vol. 433, no. 7023, pp. 312--6, Jan. 2005.

\bibitem{Ohtsuki2006}
H.~Ohtsuki, C.~Hauert, E.~Lieberman, and M.~A. Nowak, ``{A simple rule for the
  evolution of cooperation on graphs and social networks.}'' \emph{Nature},
  vol. 441, no. 7092, pp. 502--5, May 2006.

\bibitem{Nowak2004a}
M.~A. Nowak, A.~Sasaki, C.~Taylor, and D.~Fudenberg, ``{Emergence of
  cooperation and evolutionary stability in finite populations.}''
  \emph{Nature}, vol. 428, no. 6983, pp. 646--50, Apr. 2004.

\bibitem{Santos2008}
F.~C. Santos, M.~D. Santos, and J.~M. Pacheco,
  ``\BIBforeignlanguage{en}{{Social diversity promotes the emergence of
  cooperation in public goods games.}}''
  \emph{\BIBforeignlanguage{en}{Nature}}, vol. 454, no. 7201, pp. 213--6, Jul.
  2008.

\bibitem{Barabasi1999}
A.-L. Barab\'{a}si and R.~Albert, ``{Emergence of scaling in random
  networks},'' \emph{Science (New York, N.Y.)}, vol. 286, no. 5439, pp.
  509--12, Oct. 1999.

\bibitem{Meloni2009}
S.~Meloni, A.~Buscarino, L.~Fortuna, M.~Frasca, J.~G\'{o}mez-Garde\~{n}es,
  V.~Latora, and Y.~Moreno, ``{Effects of mobility in a population of
  prisoner’s dilemma players},'' \emph{Physical Review E}, vol.~79, no.~6,
  Jun. 2009.

\bibitem{Roca2011}
C.~P. Roca and D.~Helbing, ``{Emergence of social cohesion in a model society
  of greedy, mobile individuals.}'' \emph{Proceedings of the National Academy
  of Sciences of the United States of America}, vol. 108, no.~28, pp.
  11\,370--4, Jul. 2011.

\bibitem{NgPunishment}
S.-K. Ng and W.~Seah, ``Game-theoretic approach for improving cooperation in
  wireless multihop networks,'' \emph{Systems, Man, and Cybernetics, Part B:
  Cybernetics, IEEE Transactions on}, vol.~40, no.~3, pp. 559 --574, june 2010.

\bibitem{ChenCoalGm}
T.~Chen, L.~Zhu, F.~Wu, and S.~Zhong, ``Stimulating cooperation in vehicular ad
  hoc networks: A coalitional game theoretic approach,'' \emph{Vehicular
  Technology, IEEE Transactions on}, vol.~60, no.~2, pp. 566 --579, feb. 2011.

\bibitem{schwartz2011analysis}
R.~Schwartz, A.~Ohazulike, H.~van Dijk, and H.~Scholten, ``Analysis of
  utility-based data dissemination approaches in vanets,'' in \emph{Vehicular
  Technology Conference (VTC Fall), 2011 IEEE}.\hskip 1em plus 0.5em minus
  0.4em\relax IEEE, 2011, pp. 1--5.

\bibitem{schwartz2012achieving}
R.~Schwartz, A.~Ohazulike, and H.~Scholten, ``Achieving data utility fairness
  in periodic dissemination for vanets,'' in \emph{Vehicular Technology
  Conference (VTC Spring), 2012 IEEE 75th}.\hskip 1em plus 0.5em minus
  0.4em\relax IEEE, 2012, pp. 1--5.

\bibitem{shrestha2008wireless}
B.~Shrestha, D.~Niyato, Z.~Han, and E.~Hossain, ``Wireless access in vehicular
  environments using bittorrent and bargaining,'' in \emph{Global
  Telecommunications Conference, 2008. IEEE GLOBECOM 2008. IEEE}.\hskip 1em
  plus 0.5em minus 0.4em\relax IEEE, 2008, pp. 1--5.

\bibitem{hardin2009tragedy}
G.~Hardin, ``The tragedy of the commons∗,'' \emph{Journal of Natural
  Resources Policy Research}, vol.~1, no.~3, pp. 243--253, 2009.

\bibitem{hauert2004dynamics}
C.~Hauert, N.~Haiden, and K.~Sigmund, ``The dynamics of public goods,''
  \emph{Discrete and Continuous Dynamical Systems Series B}, vol.~4, pp.
  575--588, 2004.

\bibitem{hauert2006evolutionary}
C.~Hauert, M.~Holmes, and M.~Doebeli, ``Evolutionary games and population
  dynamics: maintenance of cooperation in public goods games,''
  \emph{Proceedings of the Royal Society B: Biological Sciences}, vol. 273, no.
  1600, pp. 2565--2571, 2006.

\bibitem{cardillo2012velocity}
A.~Cardillo, S.~Meloni, J.~G{\'o}mez-Gardenes, and Y.~Moreno,
  ``Velocity-enhanced cooperation of moving agents playing public goods
  games,'' \emph{Physical Review E}, vol.~85, no.~6, p. 067101, 2012.

\bibitem{wieselthier2000construction}
J.~Wieselthier, G.~Nguyen, and A.~Ephremides, ``On the construction of
  energy-efficient broadcast and multicast trees in wireless networks,'' in
  \emph{INFOCOM 2000. Nineteenth Annual Joint Conference of the IEEE Computer
  and Communications Societies. Proceedings. IEEE}, vol.~2.\hskip 1em plus
  0.5em minus 0.4em\relax IEEE, 2000, pp. 585--594.

\bibitem{tembine2009evolutionary}
H.~Tembine, ``Evolutionary network formation games and fuzzy coalition in
  heterogeneous networks,'' in \emph{Wireless Days (WD), 2009 2nd IFIP}.\hskip
  1em plus 0.5em minus 0.4em\relax IEEE, 2009, pp. 1--5.

\bibitem{tembine2010evolutionary}
H.~Tembine, E.~Altman, R.~El-Azouzi, and Y.~Hayel, ``Evolutionary games in
  wireless networks,'' \emph{Systems, Man, and Cybernetics, Part B:
  Cybernetics, IEEE Transactions on}, vol.~40, no.~3, pp. 634--646, 2010.

\bibitem{khan2011evolutionary}
M.~Khan and H.~Tembine, ``Evolutionary coalitional games in network
  selection,'' in \emph{Wireless Advanced (WiAd), 2011}.\hskip 1em plus 0.5em
  minus 0.4em\relax IEEE, 2011, pp. 185--194.

\bibitem{tembine2011dynamic}
H.~Tembine and A.~Azad, ``Dynamic routing games: An evolutionary game theoretic
  approach,'' in \emph{Decision and Control and European Control Conference
  (CDC-ECC), 2011 50th IEEE Conference on}.\hskip 1em plus 0.5em minus
  0.4em\relax IEEE, 2011, pp. 4516--4521.

\bibitem{cheng2011ecology}
S.~Cheng, P.~Chen, and K.~Chen, ``Ecology of cognitive radio ad hoc networks,''
  \emph{Communications Letters, IEEE}, vol.~15, no.~7, pp. 764--766, 2011.

\bibitem{sasabe2007caching}
M.~Sasabe, N.~Wakamiya, and M.~Murata, ``A caching algorithm using evolutionary
  game theory in a file-sharing system,'' in \emph{Computers and
  Communications, 2007. ISCC 2007. 12th IEEE Symposium on}.\hskip 1em plus
  0.5em minus 0.4em\relax IEEE, 2007, pp. 631--636.

\bibitem{sasabe2010user}
------, ``User selfishness vs. file availability in p2p file-sharing systems:
  Evolutionary game theoretic approach,'' \emph{Peer-to-peer networking and
  applications}, vol.~3, no.~1, pp. 17--26, 2010.

\bibitem{matsuda2010evolutionary}
Y.~Matsuda, M.~Sasabe, and T.~Takine, ``Evolutionary game theory-based
  evaluation of p2p file-sharing systems in heterogeneous environments,''
  \emph{International Journal of Digital Multimedia Broadcasting}, vol. 2010,
  2010.

\bibitem{Aviles1999}
L.~Avil\'{e}s, ``{Cooperation and non-linear dynamics: An ecological
  perspective on the evolution of sociality},'' \emph{Evolutionary Ecology
  Research}, vol.~4, pp. 459--477, 1999.

\bibitem{Kun2006}
A.~Kun, ``{Asynchronous snowdrift game with synergistic effect as a model of
  cooperation},'' \emph{Behavioral Ecology}, vol.~17, no.~4, pp. 633--641, Apr.
  2006.

\bibitem{Banerjee2012}
A.~Banerjee, C.~{Heng Foh}, C.~{Kiat Yeo}, and B.~{Sung Lee}, ``{Performance
  improvements for network-wide broadcast with instantaneous network
  information},'' \emph{Journal of Network and Computer Applications}, vol.~35,
  no.~3, pp. 1162--1174, May 2012.

\bibitem{Grossglauser2002}
M.~Grossglauser and D.~Tse, ``{Mobility increases the capacity of ad hoc
  wireless networks},'' \emph{IEEE/ACM Transactions on Networking}, vol.~10,
  no.~4, pp. 477--486, Aug. 2002.

\bibitem{Ioannidis2009}
S.~Ioannidis, A.~Chaintreau, and L.~Massoulie, ``{Optimal and Scalable
  Distribution of Content Updates over a Mobile Social Network},'' in
  \emph{IEEE INFOCOM 2009 - The 28th Conference on Computer
  Communications}.\hskip 1em plus 0.5em minus 0.4em\relax IEEE, Apr. 2009, pp.
  1422--1430.

\bibitem{Malandrino2011}
F.~Malandrino, C.~Casetti, C.-F. Chiasserini, and M.~Fiore, ``{Content
  downloading in vehicular networks: What really matters},'' in \emph{2011
  Proceedings IEEE INFOCOM}.\hskip 1em plus 0.5em minus 0.4em\relax IEEE, Apr.
  2011, pp. 426--430.

\bibitem{DiFelice2011}
M.~{Di Felice}, L.~Bedogni, and L.~Bononi, ``{Dynamic backbone for fast
  information delivery in vehicular ad hoc networks},'' in \emph{Proceedings of
  the 8th ACM Symposium on Performance evaluation of wireless ad hoc, sensor,
  and ubiquitous networks - PE-WASUN '11}.\hskip 1em plus 0.5em minus
  0.4em\relax New York, New York, USA: ACM Press, Nov. 2011, p.~1.

\bibitem{Li2012a}
Y.~Li, Z.~Wang, D.~Jin, L.~Zeng, and S.~Chen, ``{Collaborative Vehicular
  Content Dissemination with Directional Antennas},'' \emph{IEEE Transactions
  on Wireless Communications}, vol.~11, no.~4, pp. 1301--1306, Apr. 2012.

\bibitem{Rhee2008}
I.~Rhee, M.~Shin, S.~Hong, K.~Lee, and S.~Chong, ``{On the Levy-Walk Nature of
  Human Mobility},'' in \emph{2008 IEEE INFOCOM - The 27th Conference on
  Computer Communications}.\hskip 1em plus 0.5em minus 0.4em\relax IEEE, Apr.
  2008, pp. 924--932.

\bibitem{Lee2012}
K.~Lee, S.~Hong, S.~J. Kim, I.~Rhee, and S.~Chong, ``{SLAW: Self-Similar
  Least-Action Human Walk},'' \emph{IEEE/ACM Transactions on Networking},
  vol.~20, no.~2, pp. 515--529, Apr. 2012.

\bibitem{Li2012}
Y.~Li, D.~Jin, P.~Hui, L.~Su, and L.~Zeng, ``{Revealing contact interval
  patterns in large scale urban vehicular ad hoc networks},'' in
  \emph{Proceedings of the ACM SIGCOMM 2012 conference on Applications,
  technologies, architectures, and protocols for computer communication -
  SIGCOMM '12}.\hskip 1em plus 0.5em minus 0.4em\relax New York, New York, USA:
  ACM Press, Aug. 2012, p. 299.

\bibitem{Jiang2009}
B.~Jiang, J.~Yin, and S.~Zhao, ``{Characterizing the human mobility pattern in
  a large street network},'' \emph{Physical Review E}, vol.~80, no.~2, Aug.
  2009.

\bibitem{Barthelemy2011}
M.~Barth\'{e}lemy, ``{Spatial networks},'' \emph{Physics Reports}, vol. 499,
  no. 1-3, pp. 1--101, Feb. 2011.

\bibitem{Masucci2009}
A.~P. Masucci, D.~Smith, A.~Crooks, and M.~Batty, ``{Random planar graphs and
  the London street network},'' \emph{The European Physical Journal B},
  vol.~71, no.~2, pp. 259--271, Aug. 2009.

\bibitem{git}
\BIBentryALTinterwordspacing
``{Git repository of the simulation code}.'' [Online]. Available:
  \url{https://github.com/ComplexSystemTelecomSudParis/CooperativeNetworking}
\BIBentrySTDinterwordspacing

\end{thebibliography}
\end{document}